\documentclass[11pt]{article}
\pdfoutput=1

\usepackage{ltexpprt}
\usepackage{hyperref}
\usepackage[T1]{fontenc}
\usepackage{euscript}
\usepackage{graphicx}
\usepackage[utf8]{inputenc}
\usepackage{wrapfig}
\usepackage{xspace}
\usepackage{float}
\usepackage{marginnote}
\usepackage{enumerate}
\usepackage{enumitem}
\usepackage{epstopdf}
\usepackage[dvipsnames]{xcolor}
\usepackage[mathlines]{lineno}
\usepackage{algorithm}
\usepackage{algorithmic}
\usepackage{amsmath}
\usepackage[margin=1.25in]{geometry}
\usepackage{mathpazo}
\usepackage{calrsfs}
\usepackage{cite}

\usepackage{amsmath,amsfonts,amssymb, bm}
\usepackage{thm-restate}
\usepackage{yfonts}

\DeclareMathAlphabet{\mathpgoth}{OT1}{pgoth}{m}{n}

\def\mcA{{\mathcal{A}}}

\def\mcG{{\mathcal{G}}}

\def\mcM{{\mathcal{M}}}

\def\mcV{{\mathcal{V}}}

\def\mbG{{\mathbb{G}}}

\def\mbR{{\mathbb{R}}}

\def\mbW{{\mathbb{W}}}
\def\mbX{{\mathbb{X}}}

\def\mbZ{{\mathbb{Z}}}

\newcommand{\ignore}[1]{}

\newcommand{\eps}{\varepsilon}
\newcommand{\real}{\mathbb{R}}

\newcommand{\expect}[1]{\mathbb{E}\left[#1\right]}
\newcommand{\prob}[1]{\mathrm{Pr}\left[#1\right]}

\newcommand{\eucl}[2]{\ensuremath{\| #1 - #2 \|}}

\newcommand{\plancost}{\text{\textcent}}
\newcommand{\opt}{w^*}

\newcommand{\subcell}{\xi}
\newcommand{\cell}{\Box}
\newcommand{\subcelldiam}[1]{\delta_{#1}}

\newcommand{\subcells}[1]{\mathsf{S}[#1]}

\newcommand{\children}[1]{\mathsf{C}[#1]}

\newcommand{\grid}{\mathbb{G}}
\newcommand{\level}[1]{\mathrm{lev}(#1)}

\newcommand{\tree}{\mathcal{T}}
\newcommand{\cellsoflevel}[1]{\mathcal{L}[#1]}

\newcommand{\centerof}[1]{c_{#1}}
\newcommand{\sidelength}[1]{\ell_{#1}}

\newcommand{\spanner}{\mathcal{G}}
\newcommand{\spannercell}[1]{\mathcal{S}_{#1}}
\newcommand{\spannercellp}[1]{\mathcal{S}'_{#1}}
\newcommand{\lengthof}[1]{\phi(#1)}

\newcommand{\pathof}[1]{P_{#1}}
\newcommand{\pathofcell}[2]{P_{#1}^{#2}}

\newcommand{\localinstance}{\mathcal{I}}

\newcommand{\slack}{s}

\newcommand{\edgecost}[1]{|#1|}

\newcommand{\greedy}{\textsc{Greedy}}

\newcommand{\oracle}{\textsc{Oracle}}

\newcommand{\neighborhood}[2]{\mathcal{N}_{#1} (#2)}

\newcommand{\discrete}[1]{\hat{#1}}

\newcommand{\map}[2]{\mcM_{#1}(#2)}

\newcommand{\partition}{\mcG}
\newcommand{\rep}[1]{r_{#1}}

\newcommand{\floor}[1]{\left\lfloor #1\right\rfloor}
\newcommand{\ceil}[1]{\left\lceil #1\right\rceil}

\newcommand{\ball}[2]{D(#1, #2)}

\newcommand{\pdot}{\textsc{PD-OT}}
\newcommand{\vor}{\mathrm{Vor}}
\newcommand{\vd}{\mathrm{VD}}
\newcommand{\disc}[1]{\mathpgoth{A}_{#1}}
\newcommand{\distance}{\mathrm{d}}

\begin{document}

\newcommand*\samethanks[1][\value{footnote}]{\footnotemark[#1]}

\title{\Large Fast and Accurate Approximations of the Optimal Transport in Semi-Discrete and Discrete Settings}
\author{Pankaj K. Agarwal\thanks{Department of Computer Science, Duke University.}
\and Sharath Raghvendra\thanks{Department of Computer Science, North Carolina State University.}
\and Pouyan Shirzadian\thanks{Department of Computer Science, Virginia Tech.}
\and Keegan Yao\samethanks[1]}

\date{}

\maketitle

\begin{abstract} \small\baselineskip=9pt Given a $d$-dimensional continuous (resp. discrete) probability distribution $\mu$ and a discrete distribution $\nu$, the semi-discrete (resp. discrete) Optimal Transport (OT) problem asks for computing a minimum-cost plan to transport mass from $\mu$ to $\nu$; we assume $n$ to be the number of points in the support of the discrete distributions. 
In this paper, we present three approximation algorithms for the OT problem with strong theoretical guarantees.

\begin{enumerate}
    \item[(i)] {\it Additive approximation for semi-discrete OT:} \ignore{Our first algorithm takes a parameter $\varepsilon >0$ and}For any parameter $\varepsilon>0$, we present an algorithm that computes a semi-discrete transport plan $\tau$ with cost $\plancost(\tau) \le \plancost(\tau^*)+\varepsilon$ in  $n^{O(d)}\log\frac{\Delta}{\varepsilon}$ time; here, $\tau^*$ is the optimal transport plan, $\Delta$ is the diameter of the supports of $\mu$ and $\nu$, and we assume we have access to an oracle that outputs the mass of $\mu$ inside a constant-complexity region in $O(1)$ time. Our algorithm works for several ground distances including the $L_p$-norm and the squared-Euclidean distance.
    \item[(ii)] {\it Relative approximation for semi-discrete OT:} For any parameter $\varepsilon >0$, \ignore{assuming we have access to an oracle that outputs the mass of $\mu$ inside an orthogonal box in $O(1)$ time, }we present an algorithm that computes a semi-discrete transport plan $\tau$ with cost $\plancost(\tau)\le (1+\eps)\plancost(\tau^*)$\ignore{ under any $L_p$ norm} in $n\varepsilon^{-O(d)}\log (n)\log^{O(d)}(\log n)$ time; here, $\tau^*$ is the optimal transport plan, and we assume we have access to an oracle that outputs the mass of $\mu$ inside an orthogonal box in $O(1)$ time, and the ground distance is any $L_p$ norm.
    \item[(iii)] {\it Relative approximation for discrete OT:} For any parameter $\varepsilon > 0$, we present a Monte-Carlo algorithm that computes a transport plan $\sigma$ with an expected cost $\plancost(\sigma)\le (1+\varepsilon)\plancost(\sigma^*)$ under any $L_p$ norm in $n \varepsilon^{-O(d)} \log (n) \log^{O(d)} (\log n)$ time\ignore{ assuming that the spread of $A\cup B$ is polynomially bounded}; here, $\sigma^*$ is an optimal discrete transport plan and we assume that the spread of the supports of $\mu$ and $\nu$ is polynomially bounded.
\end{enumerate}

\end{abstract}

\section{Introduction}
\label{sec:introduction}

Optimal transport (OT) is a powerful tool for comparing probability distributions and computing maps
between them. Put simply, the optimal transport problem deforms one distribution to the other with smallest possible cost. Classically, the OT problem has been extensively studied within the operations research, statistics, and mathematics~\cite{mirebeau2015discretization, oliker1989numerical, villani2009optimal}. In recent years, optimal transport has seen rapid rise in various machine learning and computer vision applications as a meaningful metric between distributions and has been 
extensively used in generative models~\cite{deshpande2018generative, genevay, salimans}, robust learning~\cite{esfahani}, supervised learning~\cite{janati2019wasserstein, luise2018differential}, computer vision applications~\cite{indykicml, gupta2010sparse}, variational inference~\cite{ambrogioni2018wasserstein}, blue noise generation~\cite{de2012blue, qin2017wasserstein}, and parameter estimation~\cite{bernton, generative}. These applications have led to developing efficient algorithms for OT; see the book~\cite{peyre2019computational} for review of computational OT. 

In the \emph{geometric OT} problem, the cost of transporting unit mass between two locations is the Euclidean distance or some $L_p$ norm between them. In this paper, we design simple, efficient approximation algorithms 
for the semi-discrete and discrete geometric OT problems in fixed dimensions.

Let $\mu$ be a continuous probability distribution (i.e., density) defined over a compact bounded support $A\subset \mbR^d$, and let $\nu$ be a discrete distribution, where the support of $\nu$, denoted by $B$, is a set of $n$ points in $\mbR^d$. Let $d(\cdot, \cdot)$ be the ground metric between a pair of points in $\mathbb{R}^d$. A coupling $\tau \colon A \times B \to \mathbb{R}_{\geq 0}$ is called a \emph{transport plan} for $\mu$ and $\nu$ if for all $a \subseteq A$, $\sum_{b \in B} \tau(a, b) = \mu(a)$ (where $\mu(a)$ is the mass of $\mu$ inside $a$) and for all $b \in B$,  $\int_{A} \tau(a, b) \, da = \nu(b)$. 
The cost of the transport plan $\tau$ is given by $\plancost(\tau) := \int_{A} \sum_{b \in B} d(a,b) \tau(a, b) \, da$. The goal is to find a minimum-cost (semi-discrete) transport plan satisfying $\mu$ and $\nu$\footnote{Apparently the semi-discrete OT was introduced by Cullen and Purser \cite{cullen1984extended} without reference to optimal transport.}. For any parameter $\varepsilon>0$, a transport plan $\tau$ between $\mu$ and $\nu$ is called \emph{$\varepsilon$-close} if the cost of $\tau$ is within an additive error of $\varepsilon$ from the cost of the optimal transport plan $\tau^*$, i.e., $\plancost(\tau) \le \plancost(\tau^*) + \varepsilon$.  A \emph{$(1+\eps)$-approximate OT plan}, or simply \emph{$\eps$-OT plan}, is a transport plan $\tau$ with $\plancost(\tau)\le(1+\eps)\plancost(\tau^*)$. 

The problem of computing semi-discrete optimal transport between $\mu$ and $\nu$ reduces to the problem of finding a set of weights $y:B\rightarrow \mbR_{\ge 0}$ so that, for any point $b \in B$, the Voronoi cell of $b$ in the additively weighted Voronoi diagram has a mass equal to $\nu(b)$, i.e., $Vor(b)=\{x\in\mbR^d\mid d(x,b)-y(b)\le d(x,b')-y(b'), \forall b'\in B\}, \; \mu(Vor(b))=\nu(b)$, and the mass of $\mu$ in $Vor(b)$ is transported to $b$; see \cite{aurenhammer1992minkowski}. One can thus define an optimal semi-discrete transport plan by describing the weights of points in $B$. For arbitrary distributions, weights can have large bit (or algebraic) complexity, so our goal will be to compute the weights accurately up to $s=O(\log \varepsilon^{-1})$ bits, which in turn will return an $\varepsilon$-close semi-discrete OT plan.

If $\mu$ is also a discrete distribution with support $A$, a \emph{discrete transport plan} is $\sigma \colon A \times B \to \real_{\geq 0}$ that assigns the mass transported along each edge $(a,b)\in A\times B$ such that $\sum_{b\in B} \sigma(a,b) = \mu(a)$ for each point $a\in A$ and $\sum_{a \in A}\sigma(a,b)=\nu(b)$ for each point $b\in B$. The cost of $\sigma$ is given by $\plancost(\sigma) = \sum_{(a,b) \in A\times B}\sigma(a,b)d(a,b)$.
The \emph{discrete OT problem} asks for a transport plan $\sigma$ with the minimum cost. We refer to such plan as an \emph{OT plan}.

\subsubsection*{Related work.} The discrete optimal transport problem under any metric can be modeled as an uncapacitated minimum-cost flow problem and can be solved in strongly polynomial time of $O((m + n\log n) n\log n)$ time using the algorithm by Orlin~\cite{orlin1988faster}. Using recent techniques~\cite{seshadri2014algorithm}, it can be solved in $n^{2+o(1)}\mathrm{poly}\log(\Delta)$ time, where $\Delta$ depends on the spread of $A\cup B$ and the maximum demand. The special case where all points have the same demand is the widely studied \textit{minimum-cost bipartite matching} problem. There is extensive work on the design of near-linear time approximation for the optimal transport and related matching problems~\cite{as-stoc-14,andoni2008earth, indykicml, fox-lu-transport, khesin2019preconditioning, raghvendra2020near, ra_soda12}.  The near-linear time algorithms by Khesin~\textit{et. al.}~\cite{khesin2019preconditioning} and Fox and Lu~\cite{fox-lu-transport} for computing an $\eps$-OT plan use minimum-cost-flow (MCF) solvers (e.g. \cite{sherman}) as a black box and numerically precondition their minimum-cost flow instance using geometry~\cite{fox-lu-transport, khesin2019preconditioning, sherman2017generalized}. The work of Zuzic~\cite{zuzic2021simple} describes a multiplicative-weights update (MWU) based boosting method for minimum-cost flows using an approximate primal-dual oracle as a black box, which replaces the preconditioner used in~\cite{khesin2019preconditioning, sherman2017generalized}. All these algorithms are Monte Carlo algorithms and have running time of $n(\varepsilon^{-1}\log n)^{O(d)}$. Recently, Agarwal \textit{et. al.}~\cite{agarwal2022deterministic} presented an $n(\eps^{-1}\log n)^{O(d)}$-time deterministic algorithm for computing an $\eps$-approximate bipartite matching in $\real^d$. A Monte-Carlo $\varepsilon$-approximation algorithm for matching with run time $n\log^4 n(\varepsilon^{-1}\log\log n)^{O(d)}$ was presented in~\cite{agarwal2022improved}. Very recently, Fox and Lu proposed a deterministic algorithm for $\varepsilon$-OT with run time of $O(n\varepsilon^{-(d+2)} \log^5 n \log\log n)$ \cite{fox2022deterministic}.

The known algorithms for semi-discrete OT that compute an $\varepsilon$-close transport plan by and large use first and second order numerical solvers \cite{aurenhammer1992minkowski, benamou2000computational, chartrand2009gradient, de2012blue, kitagawa2014iterative, kitagawa2019convergence, levy2018notions, oliker1989numerical}. These algorithms start with an initial set of weights for points in $B$ and iteratively improve the weights until the mass inside the Voronoi cell of any point $b \in B$ is an additive factor $\varepsilon$ away from $\nu(b)$. One can use these solvers to compute an $\varepsilon$-close transport plan by executing $\mathrm{poly}(n,1/\varepsilon)$ iterations. Each iteration requires computation of several weighted Voronoi diagrams which takes $n^{\Omega(d)}$ time. 
One can also draw samples from the continuous distribution and convert the semi-discrete OT problem to a discrete instance~\cite{genevay2016stochastic}; however, due to sampling errors, this approach provides an additive approximation. 
Van Kreveld \textit{et. al.}~\cite{van2021approximating} presented a $(1+\varepsilon)$-approximation OT algorithm for the restricted case when the continuous distribution is uniform over a collection of simple geometric objects (e.g. segments, simplices, etc.), by sampling roughly $n^2$ points and then running an algorithm for computing discrete $\varepsilon$-OT mentioned above. Their running time is roughly $n^2 \varepsilon^{-O(d)}\mathrm{poly}\log(n)$.

\subsubsection*{Our contributions.}
We present three new algorithms for the semi-discrete and discrete optimal transport problems. 
Our first result is a cost-scaling algorithm that computes an $\varepsilon$-close transport plan for a semi-discrete instance in $n^{O(d)}\log (\Delta/\varepsilon)$ time, assuming that we have access to an oracle that, given a constant complexity region $\varphi$, returns $\mu(\varphi)$. 

\begin{theorem} \label{theorem:semi-continuous-exact}
Let $\mu$ be a continuous distribution defined on a compact bounded set $A\subset \mbR^d$, $\nu$ a discrete distribution with a support $B\subset \mbR^d$ of size $n$, and $\varepsilon >0$ a parameter. Suppose there exists an $\oracle$ which, given a constant complexity region $\varphi$, returns $\mu(\varphi)$ in $Q$ time. Then, an $\varepsilon$-close semi-discrete OT plan can be computed in $Qn^{O(d)}\log (\frac{\Delta}{\varepsilon})$ time, where $\Delta$ is the diameter of $A\cup B$.  
\end{theorem}

To the best of our knowledge, our algorithm is the first one to compute an $\varepsilon$-close OT in time that is polynomial in both $n$ and $\log (\varepsilon^{-1})$. Earlier algorithms had an $\varepsilon^{-O(1)}$ factor in the run time\footnote{M\'erigot and Thibert had conjectured that an algorithm for computing an $\varepsilon$-close OT for semi-discrete setting with runtime $(n\log\eps^{-1})^{O(1)}$ might follow using a scaling framework~\cite[Remark 24]{merigot2021optimal}. Our result proves their conjecture in the affirmative.}. Our algorithm not only finds the optimal transport cost within an additive error, it also finds the optimal dual weights within an additive error of $\varepsilon$, i.e., it computes optimal dual-weights up to $O(\log \varepsilon^{-1})$ bits of accuracy.
Our algorithm works for any ground distance where the bisector of two points under the distance function $d(\cdot ,\cdot)$ is an algebraic variety of constant degree. Consequently, it works for several important distances, including the $L_p$-norm and the squared-Euclidean distance.

The previous best-known algorithm by Kitagawa~\cite{kitagawa2014iterative} for the semi-discrete optimal transport has an execution time $n^{\Omega(d)}\Delta/\varepsilon$; furthermore, their algorithm only approximates the cost and does not necessarily provide any guarantees for the optimal transport plan or the optimal dual weights of $B$.

For each scale $\delta$, our algorithm starts with a set of weights assigned to $B$. Using these weights, it constructs an instance of the discrete optimal transport of size $n^{O(d)}$, which is then solved using a primal-dual solver. The optimal dual weights for this discrete instance are then used to refine the dual weights of $B$. These refined dual weights act as the starting dual weights for the next scale $\delta/2$. Starting with $\delta = \Delta$, our algorithm executes a total of $O(\log (\Delta/\varepsilon))$ scales.

Our main insight is that in scale $\delta$, one can partition the continuous distribution $\mu$ into exponentially many regions $A_\delta$.   We prove that the dual weights and the semi-discrete transport plan $\tau$ computed by our algorithm satisfy a set of $\delta$-optimal dual feasibility conditions (a relaxation of the classical feasibility conditions of the optimal transport), one for each $(\varrho,b) \in A_\delta \times B$, making $\tau$ a $\delta$-close transport plan. Unfortunately, explicitly solving for $\tau$ using the partitioning $A_\delta$ will result in an exponential execution time. We overcome this difficulty by making two observations. 

At the start of scale $\delta$, we have a very good initial estimate for the dual weights of points in $B$ from the ones computed in the previous scale. In particular, we show that there is a semi-discrete transport plan $\tau$ such that the dual feasibility constraints on every pair $(\varrho,b) \in A_\delta\times B$ with $\tau(\varrho,b) > 0$ has a slack $\le 4n\delta$. Using this claim, we show that in the optimal semi-discrete transport plan $\tau^*$, $\tau^*(\varrho,b)=0$ for every pair $(\varrho,b)$ with a slack $> 4n\delta$. This allows us to restrict our attention to edges with slack $\le 4n\delta$. Unfortunately, there can be exponential number of edges with slack at most $4n\delta$.
In order to overcome this difficulty, we show that all slack $i$ edges incident on $b$ can be compactly represented as regions between carefully constructed expansions of $O(n)$ Voronoi cells in the weighted Voronoi diagram. Using this property, we can compress the size of OT instance to $n^{O(d)}$, which can then be solved using a discrete OT solver.

We also show that by increasing the number of scales in our algorithm from $O(\log (\Delta/\varepsilon))$ to $O(\log (n\Delta/\varepsilon))$, we obtain the optimal weights on the points in $B$ within an additive error of $\varepsilon$.

Next, we present another approximation algorithm for the semi-discrete setting whose running time is near-linear in $n$ but the dependence on $\varepsilon$ increases to $\varepsilon^{-O(d)}$.

\begin{theorem} \label{theorem:semi-continuous}
    Let $\mu$ be a continuous distribution defined on a compact set $A\subset \mbR^d$, $\nu$ a discrete distribution with a support $B\subset \mbR^d$ of size $n$, and $\varepsilon >0$ a parameter. Suppose there exists an $\oracle$ which, given an axis-aligned box $\cell$, returns $\mu(\cell)$ in $Q$ time. Then, a $(1+\varepsilon)$-approximate semi-discrete OT plan can be computed in $O(n\varepsilon^{-3d-2} (\log^5 (n) \log (\log n) + Q))$ time. If the spread of $B$ is polynomially bounded, a $(1+\varepsilon)$-approximate semi-discrete OT plan can be computed in $O(n\varepsilon^{-4d-5} (\log (n) \log^{2d+5} (\log n) + Q))$ time with probability at least $\frac{1}{2}$.
\end{theorem}

Similar to \cite{van2021approximating}, the high level view of our approach is to discretize the continuous distribution and use a discrete OT algorithm. Our main contribution is a more clever sampling strategy that is more global and that works for arbitrary density (rather than for collections of geometric objects). We prove that it suffices to sample $n\varepsilon^{-O(d)}$ points in contrast to $\Omega(n^2)$ points in~\cite{van2021approximating}.

Our final result is a new $(1+\varepsilon)$-approximation algorithm for the discrete transport problem.

\begin{theorem} \label{theorem:OT-discrete}
    Let $\mu$ and $\nu$ be two discrete distributions with support sets $A, B\subset \mbR^d$, respectively, where $A\cup B$ is a point set of size $n$ with polynomially bounded spread, $d \geq 1$ is a constant and $\varepsilon >0$ a parameter. Then, a $(1+\varepsilon)$-approximate discrete OT plan between $\mu$ and $\nu$ can be computed by a Monte Carlo algorithm in $O\left( n \eps^{-2d-5}\log (n) \log^{2d+5}(\log n)  \right)$ time with probability at least $\frac{1}{2}$.
\end{theorem}

As mentioned above, until recently, the best-known Monte Carlo algorithm for computing an $\varepsilon$-OT plan had running time $n (\varepsilon^{-1} \log n)^{O(d)}$. Recently in an independent work, Fox and Lu \cite{fox2022deterministic} obtained a deterministic algorithm for computing an $\varepsilon$-OT plan in $O(n \varepsilon^{-d-2} \log^5(n) \log (\log n))$ time. We believe that our result is of independent interest. The running time is slightly better than in \cite{fox2022deterministic}, though of course their algorithm is deterministic. But our main contribution is a greedy primal-dual $O(\log \log n)$-approximation algorithm that is simple and geometric and runs in $O(n \log \log n)$ time. By plugging our algorithm into the multiplicative weight update method as in \cite{zuzic2021simple}, we obtain a $(1+\varepsilon)$-approximation algorithm. We believe the derandomization technique of Lu and Fox can be applied to our algorithm, but one has to check all the technical details.

\section{Computing a Highly Accurate Semi-Discrete Optimal Transport}\label{sec:scaling_algo_refined}
Given a continuous distribution $\mu$ over a compact bounded set $A\subset\mbR^d$, a discrete distribution $\nu$ over a set $B\subset \mbR^d$ of $n$ points, and a parameter $\varepsilon>0$, we present a cost-scaling algorithm for computing an $\varepsilon$-close semi-discrete transport plan from $\mu$ to $\nu$. We first describe the overall framework, then provide details of the algorithm and analyze its efficiency, and finally prove its correctness.

In our algorithm, we use a black-box primal-dual discrete OT solver $\pdot(\mu', \nu')$ that given two discrete distributions $\mu'$ and $\nu'$ defined over two point sets $A'$ and $B'$, returns a transport plan $\sigma$ from $\mu'$ to $\nu'$ and a dual weight $y(v)$ for each point $v \in A'\cup B'$ such that for any pair $(a,b) \in A'\times B'$, 
\begin{eqnarray}
    y(b)-y(a) &\le& \distance(a,b)\label{eq:exact-feasibility_non_matching},\\
    y(b)-y(a) &=& \distance(a,b) \quad \text{if } \sigma(a,b) > 0.\label{eq:exact-feasibility_matching}
\end{eqnarray}
Standard primal-dual methods~\cite{kuhn1955hungarian} construct a transport plan while maintaining~\eqref{eq:exact-feasibility_non_matching} and~\eqref{eq:exact-feasibility_matching}. For concreteness, we use Orlin's algorithm~\cite{orlin1988faster} that runs in $O(|A\cup B|^3)$ time. 

\subsection{The Scaling Framework.}\label{sec:scaling}
The algorithm works in $O(\log(\Delta\varepsilon^{-1}))$ rounds, where $\Delta$ is the diameter of $A\cup B$. In each round, we have a parameter $\delta>0$ that we refer to as the \textit{current scale}, and we also maintain a dual weight $y(b)$ for every point $b\in B$. Initially, in the beginning of the first round, $\delta = \Delta$ and $y(b)=0$ for all $b\in B$. Execute the following steps $s=c\log_2 (\Delta\varepsilon^{-1})$ times, where $c$ is a sufficiently large constant\footnote{Computing an $\varepsilon$-close transport plan requires $O(\log(\Delta/\varepsilon))$ iterations. When the goal, on the other hand, is to obtain accurate dual weights up to $O(\log \varepsilon^{-1})$ bits, we need to execute our algorithm for $O(\log (n\Delta/\varepsilon))$ iterations. See Section~\ref{sec:optimal_duals}.}.
\begin{itemize}
    \item[(i)] {\it Construct a discrete OT instance:} Using the current values of dual weights of $B$, as described below, construct a discrete distribution $\hat{\mu}_\delta$ with a support set $X_\delta$, where $|X_\delta|=n^{O(d)}$, and define a (discrete) ground distance function $\distance_\delta:B\times X_\delta \rightarrow \{0,\ldots, 4n+1\}$.
    \item[(ii)] {\it Solve OT instance:} Compute an optimal transport plan between discrete distributions $\hat{\mu}_\delta$ and $\nu$ using the procedure $\pdot(\hat{\mu}_\delta, \nu)$. Let $\sigma_\delta$ be the coupling and $\hat{y}:B\rightarrow\mbR$ be the dual weights returned by the procedure.
    \item[(iii)] {\it Update dual weights:} $y(b) \leftarrow y(b)+\delta \hat{y}(b)$ for each point $b\in B$.
    \item[(iv)] {\it Update scale:} $\delta \leftarrow \delta/2$.
\end{itemize}
We refer to the $j$th iteration of this algorithm as \emph{iteration} $j$. Our algorithm terminates when $\delta\le \varepsilon$.
We now describe the details of step (i) of our algorithm, which is the only non-trivial step. Let $y(\cdot)$ be the dual weights of $B$ at the start of iteration $j$.

\subsubsection*{Constructing a discrete OT instance.}
We construct the discrete instance by constructing a family of Voronoi diagrams and overlaying some of their cells. For a weighted point set $P\subset \mbR^d$ with weights $w:P\rightarrow\mbR$ and a distance function $\distance:P\times \mbR^d\rightarrow\mbR_{\ge0}$, we define the \emph{weighted distance} from a point $p\in P$ to any point $x\in\mbR^d$ as $\distance_w(p,x)=\distance(p,x)-w(p)$. For a point $p\in P$, its \emph{Voronoi cell} is $\vor_w(p) = \{x\in \mbR^d\mid \distance_w(p,x)\le \distance_w(p',x), \forall p'\in P\}$,
and the \emph{Voronoi diagram} $\vd_w(P)$ is the decomposition of $\mbR^d$ induced by Voronoi cells; see \cite{fortune1995voronoi}.

For $i\in[1,4n+1]$ and a point $b\in B$, we define a Voronoi cell $V_b^i$ using a weight function $w_i:B\rightarrow\mbR_{\ge0}$, as follows. We set $w_i(b)=y(b)+i\delta$ and $w_i(b')=y(b')$ for all $b'\neq b$. We set $V_b^i=\vor_{w_i}(b)$ in $\vd_{w_i}(B)$. By construction, $V_b^1\subseteq V_b^2\subseteq\ldots\subseteq V_b^{4n+1}$. Set $\mcV_b=\{V_b^i\mid i\in[1,4n+1]\}$ and $\mcV=\bigcup_{b\in B}\mcV_b$ 
(See Figure~\ref{fig:vor}(a)). Let $\mcA(\mcV)$ be the \textit{arrangement} of $\mcV$, the decomposition of $\mbR^d$ into (connected) cells induced by $\mcV$; each cell of $\mcA(\mcV)$ is the maximum connected region lying in the same subset of regions of $\mcV$~\cite{agarwal1998efficient}.

For each cell $\varphi$ in $\mcA(\mcV)$, we choose a point $\rep{\varphi}$ arbitrarily and set its mass to $\hat{\mu}_\delta(\rep{\varphi})=\mu(\varphi)$, where for any region $\rho$ in $\mbR^d$, $\mu(\rho) = \int_\rho \mu(a)\, da$ is the mass of $\mu$ inside $\rho$ (Here we assume the mass to be $0$ outside the support $A$ of $\mu$). Set $X_\delta=\{\rep{\varphi}\mid \varphi\in\mcA(\mcV)\}$. The resulting mass distribution on $X_\delta$ is $\hat{\mu}_\delta$. 

The ground distance $\distance_\delta(a,b)$ between any point $b \in B$ and a point $a \in X_\delta$ is defined as
\begin{equation*}
    \distance_\delta(a,b)=\begin{cases}
        0, \quad &\text{if } a\in V_b^1,\\
        i, \quad &\text{if } a\in V_b^{i+1}\setminus V_b^i, \ \ i\in [1,4n],\\
        4n+1,\quad &\text{if } a\notin V_b^{4n+1}.
    \end{cases}
\end{equation*}
See Figure~\ref{fig:vor}(b).  Since each $V_b^i$ is defined by $n$ algebraic surfaces of constant degree, assuming the bisector of two points under the distance function $\distance(\cdot ,\cdot)$ is an algebraic variety of constant degree, $\mcA(\mcV)$ has $n^{O(d)}$ cells and a point in every cell of $\mcA(\mcV)$ can be computed in $n^{O(d)}$ time~\cite{basu2003algorithms}. Hence, $|X_\delta|=n^{O(d)}$.
This completes the construction of $X_\delta, \hat{\mu}_\delta,$ and $\distance_\delta$.

\subsubsection*{Computing a semi-discrete transport plan.} At the end of any scale $\delta$, we compute a $\delta$-close semi-discrete transport plan $\tau_\delta$ from the discrete transport plan $\sigma_\delta$ as follows: For any edge $(\rep{\varphi},b)\in X_\delta\times B$, we arbitrarily transport $\sigma_\delta(\rep{\varphi}, b)$ mass from the points inside the region $\varphi$ to the point $b$. A simple construction of such transport plan is to set, for any region $\varphi$, any point $a\in\varphi$, and any point $b\in B$, $\tau_\delta(a,b)=\frac{\mu(a)}{\hat{\mu}_\delta(\rep{\varphi})}\sigma_\delta(\rep{\varphi}, b)$. Our algorithm will only compute the transport plan at the end of the last scale, i.e., $\delta \le \varepsilon$.

\begin{figure}
    \centering
    \begin{tabular}{c@{\hskip 5em}c}
         \includegraphics[width=0.4\textwidth]{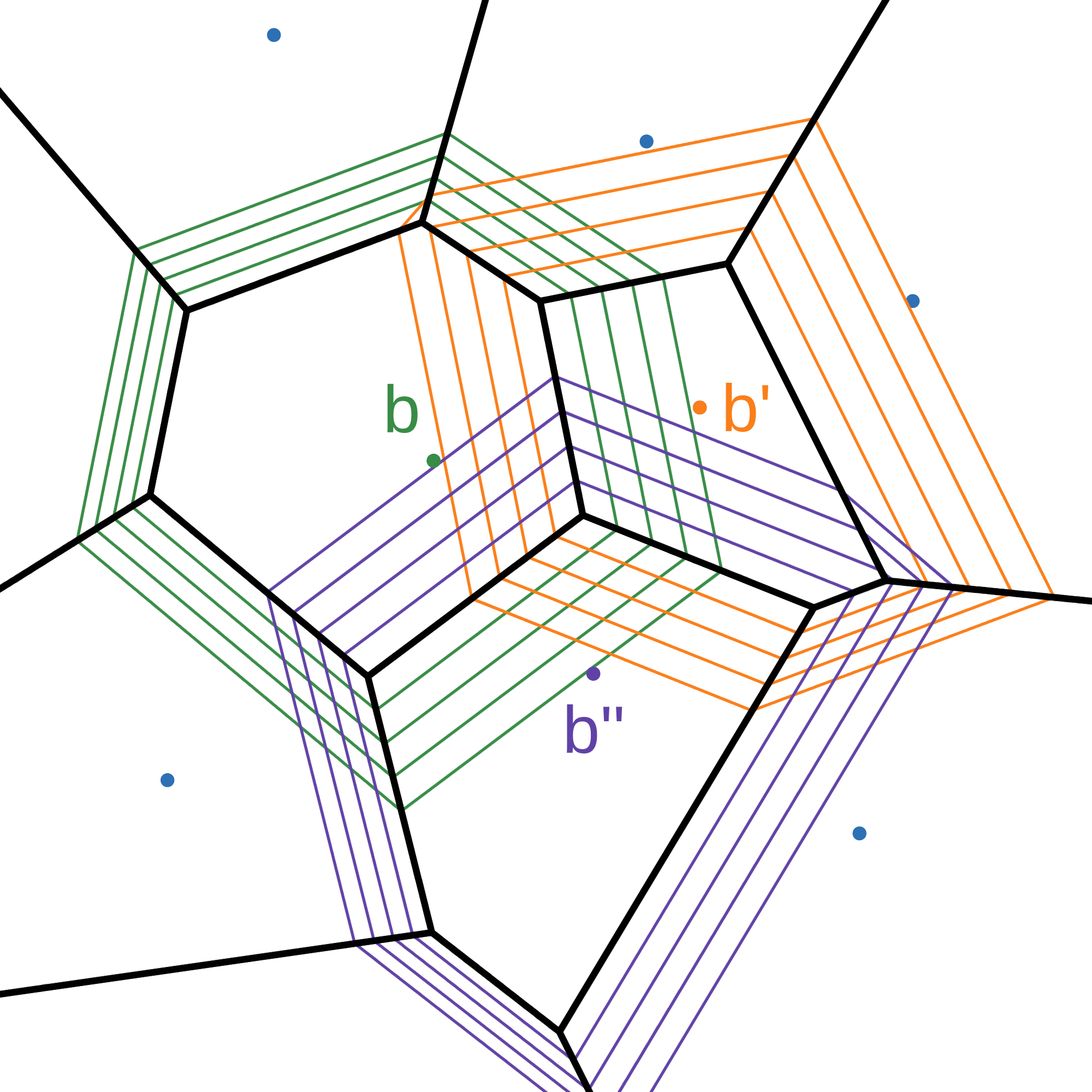} &     \includegraphics[width=0.4\textwidth]{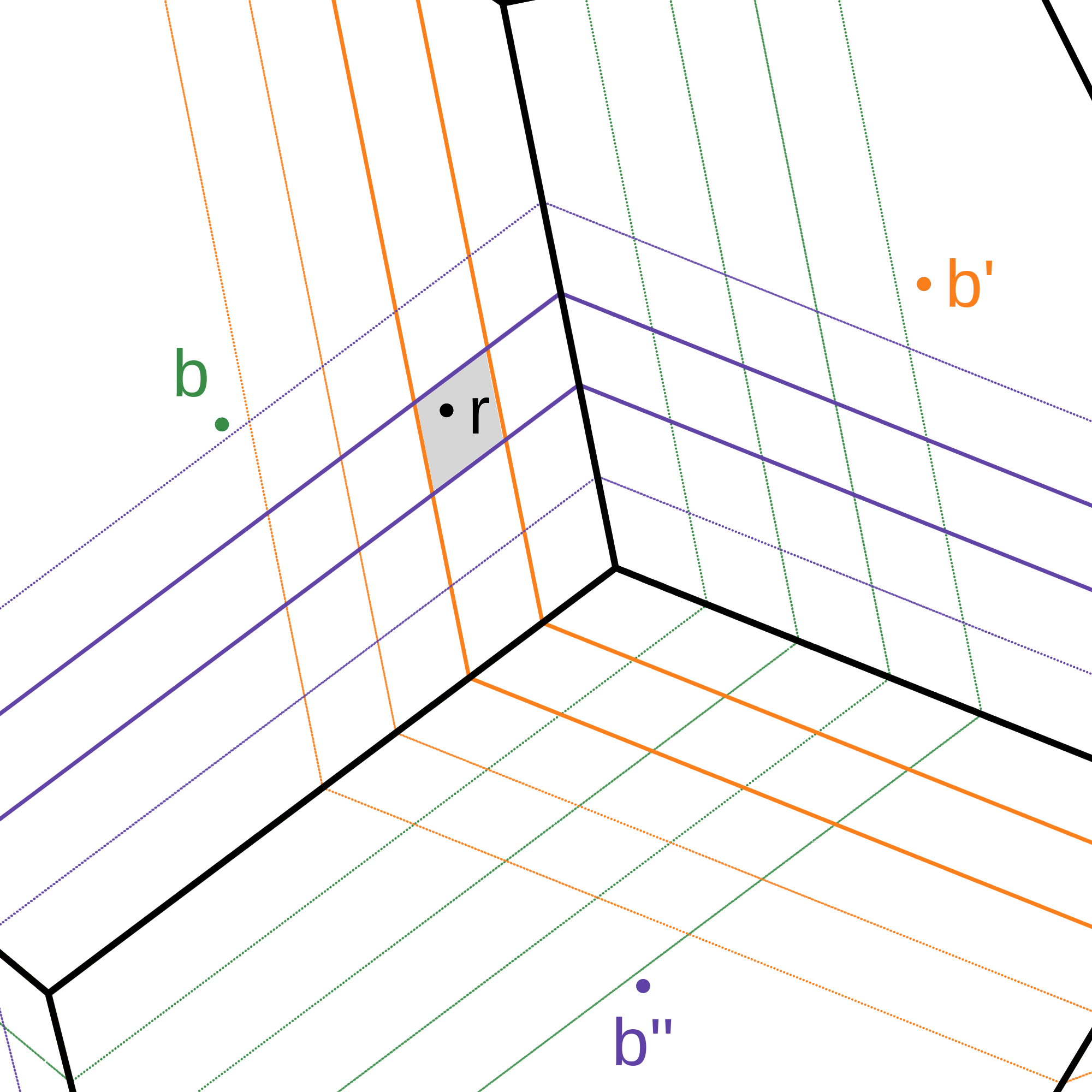}  \\
         (a) & (b) 
    \end{tabular}
    \caption{(a) The $i$-expansions of the Voronoi cells of three points $b, b', b''\in B$, (b) A region $\varphi\in\mcA(\mcV)$ (highlighted in gray) with a representative point $r\in X_\delta$, where $\distance_\delta(b,r)=0$ since $r\in V_b^1$, $\distance_\delta(r,b') = 1$ since $r\in V_{b'}^2\setminus V_{b'}^1$, and $\distance_\delta(r,b'')=2$ since $r\in V_{b''}^3\setminus V_{b''}^2$ is between the $2$-expansion and $3$-expansion of Voronoi cell of $b''$. The ground distance in this figure is squared Euclidean.}
    \label{fig:vor}
\end{figure}

\subsubsection*{Efficiency analysis.}
Our algorithm runs $O(\log (\Delta\varepsilon^{-1}))$ scales, where in each scale, it constructs a discrete OT instance in $n^{O(d)}$ time and solves the OT instance using a polynomial-time primal-dual OT solver. Since the size of the discrete OT instance is $n^{O(d)}$, solving it also takes $n^{O(d)}$ time, resulting in a total execution time of $n^{O(d)}\log(\Delta\varepsilon^{-1})$ for our algorithm.

\subsection{Proof of Correctness.}\label{sec:correctness}
In the discrete setting, cost scaling algorithms obtain an $\varepsilon$-close transport plan that satisfies~\eqref{eq:exact-feasibility_matching} and an additive $\varepsilon$ relaxation of~\eqref{eq:exact-feasibility_non_matching}. 
For our proof, we extend these relaxed feasibility conditions to the semi-discrete transport plan and show that, at the end of each scale $\delta$, the semi-discrete transport plan computed by our algorithm satisfies these conditions. We use the relaxed feasibility conditions to show that our semi-discrete transport plan is $\delta$-close. Thus, in the last scale, when $\delta \le \varepsilon$, our algorithm returns an $\varepsilon$-close semi-discrete transport plan from $\mu$ to $\nu$. 

\vspace{0.5em}
\noindent{\bf $\delta$-optimal transport plan.}
For any scale $\delta$, we first describe a discretization of the continuous distribution into a set of regions $\disc{\delta}$ and then describe the relaxed feasibility conditions for all pairs $(\varrho, b) \in \disc{\delta}\times B$. 

Consider a decomposition of the support $A$ of the continuous distribution $\mu$ into a set of regions, where each region $\varrho$ in the decomposition satisfies the following condition:\begin{itemize}
\item[(P1)]Assuming every point $b \in B$ has a weight $w(b)$ that is an integer multiple of $\delta$, any two points $x$ and $y$ in $\varrho$ have the same weighted nearest neighbor in $B$ with respect to weights $w(\cdot)$,
\end{itemize}
where for any set of weights $w$ for points in $B$ and any point $a\in A$, we say that a point $b\in B$ is a \emph{weighted nearest neighbor} of $a$ if $\distance_w(a,b)= \min_{b'\in B} \distance_w(a,b')$.
Let this set of regions be $\disc{\delta}$. For each region $\varrho\in \disc{\delta}$, let $\rep\varrho$ denote an arbitrary representative point inside $\varrho$. 

Let $y: B \rightarrow \mbR$ denote a set of dual weights for the points in $B$. For each region $\varrho\in \disc{\delta}$, we derive a dual weight $y_\delta(\rep\varrho)$ for its representative point as follows. Let $b_\varrho\in B$ be the weighted nearest neighbor of $\rep\varrho$ with respect to weights $y(\cdot)$. We set the dual weight of $\rep\varrho$ as
\begin{equation}\label{eq:dual_a_assignment}
    y_\delta(\rep\varrho) \leftarrow y(b_\varrho) - \distance(\rep\varrho,b_\varrho)-\delta.
\end{equation}
We say that a semi-discrete transport plan $\tau$ from $\mu$ to $\nu$ along with the set of dual weights $y(\cdot)$ for points in $B$ is \emph{$\delta$-optimal} if, for each point $b\in B$ and each region $\varrho\in \disc{\delta}$,
\begin{eqnarray}
    y(b)-y_\delta(\rep\varrho) &\le& \distance(\rep\varrho,b) + \delta, \label{eq:feasibility_non_matching}\\
    y(b)-y_\delta(\rep\varrho) &\ge& \distance(\rep\varrho,b) \quad\quad\ \  \text{if } \tau(\varrho,b) > 0.\label{eq:feasibility_matching}
\end{eqnarray}
In the following lemma, we show that any $\delta$-optimal transport plan $\tau, y(\cdot)$ from $\mu$ to $\nu$ is $3\delta$-close.

\begin{restatable}{lemma}{deltaoptimal}\label{lemma:delta-optimal}
    Suppose $\tau, y(\cdot)$ is any $\delta$-optimal transport plan from $\mu$ to $\nu$ and let $\tau^*$ denote any optimal transport plan from $\mu$ to $\nu$. Then, $\plancost(\tau_\sigma)\le \plancost(\tau^*)+\delta$. 
\end{restatable}

Let $y(\cdot)$ denote the set of dual weights maintained by our algorithm at the beginning of scale $\delta$.
For any point $b\in B$ and any region $\varrho\in \disc{\delta}$, we define a \emph{slack} on condition~\eqref{eq:feasibility_non_matching} for the pair $(\varrho,b)$, denoted by $s_\delta(\varrho, b)$, as
\[s_\delta(\varrho, b):=\floor{\frac{\distance(\rep\varrho,b) + \delta - y(b) + y_\delta(\rep\varrho)}{\delta}}\delta.\]

In the following, we describe the discretization of the continuous distribution into $\disc{\delta}$ and relate it to the discrete OT instance that is constructed in step (i) of our algorithm. Furthermore, we relate the distance $\distance_\delta$ computed in our algorithm to the slacks $s_\delta$.

\vspace{0.5em}
\noindent{\bf Discretizing the continuous distribution.}
Let $B=\{b_1,b_2,\ldots, b_n\}$, and let $w=\langle w_1,\ldots, w_n\rangle$ be an $n$-dimensional vector representing a weight assignment to the points in $B$. 
We say that the vector $w$ is \emph{valid} if each $w_i$ is a non-negative integer multiple of $\delta$ and bounded by $\Delta$. Consider the set $\mbW_\delta$ of all valid vectors, i.e., $\mbW_\delta=(\delta\mbZ\cap[0,\Delta])^n$. For a valid vector $w\in \mbW_\delta$, let $\vd_w(B)$ denote the weighted Voronoi diagram constructed for the points in $B$ with weights $w$. The partitioning $\disc{\delta}$ is simply the overlay of all weighted Voronoi diagrams $\vd_w(B)$ across all valid weight vectors $w\in \mbW_\delta$ (See Figure~\ref{fig:vor2}).

\begin{figure}
    \centering
    \includegraphics[width=0.4\textwidth]{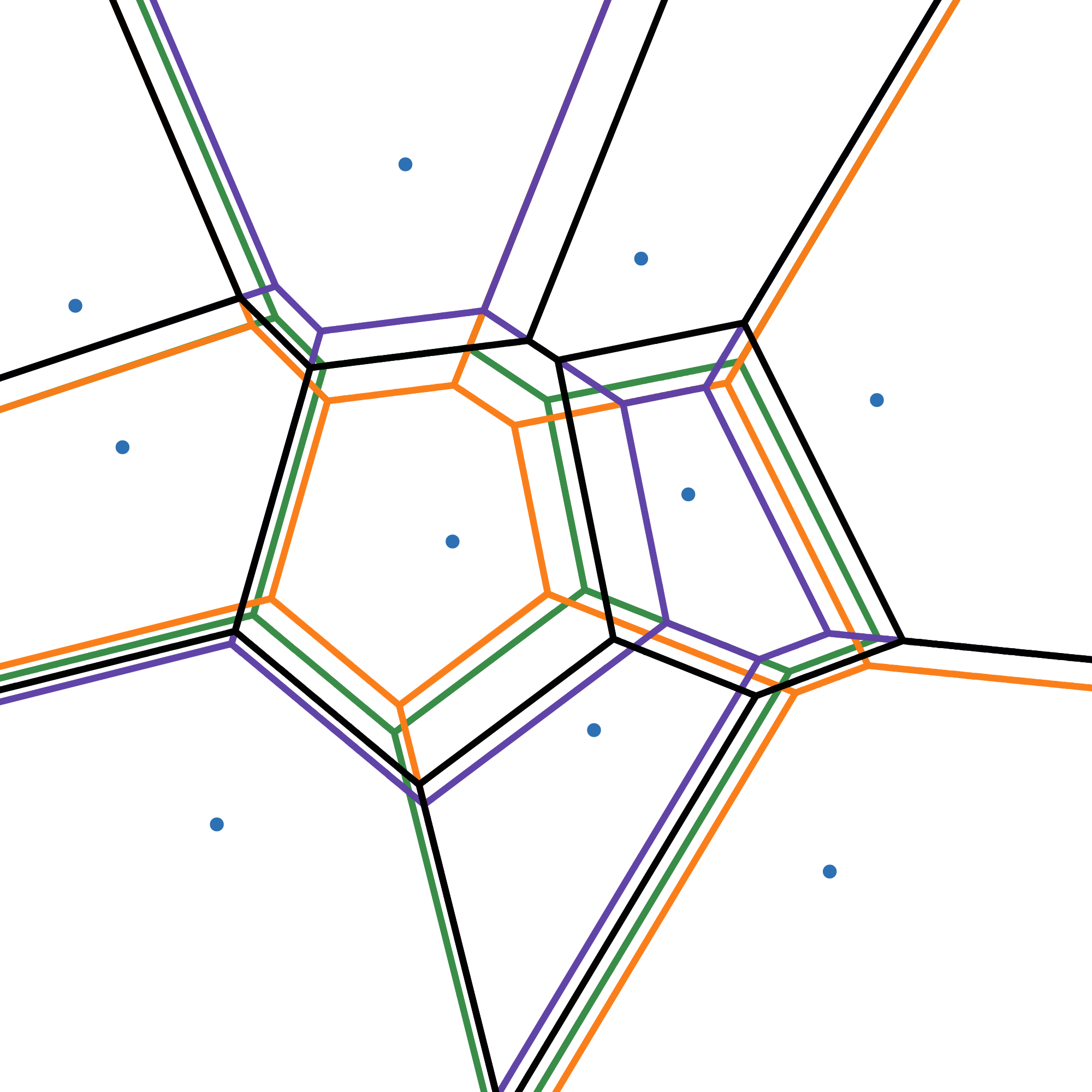} 
    \caption{The weighted Voronoi diagrams for four different weight vectors in $\mbW_\delta$. The ground distance in this figure is squared Euclidean.}
    \label{fig:vor2}
\end{figure}

At the beginning of scale $\delta$, while constructing the set $\mcA(\mcV)$, the dual weight of each point in $B$ maintained by our algorithm is obtained from scale $2\delta$ and hence, is an integer multiple of $2\delta$. Therefore, the Voronoi cells $V_b^i$ of each point $b\in B$ correspond to valid weight vectors. By construction of the set $\disc{\delta}$, each region $\varrho\in \disc{\delta}$ completely lies inside some region $\varphi\in\mcA(\mcV)$, i.e., each region in $\mcA(\mcV)$ consists of a collection of regions in $\disc{\delta}$.
In the next lemma, we establish a connection between the slacks and the distances $\distance_\delta$.

\begin{restatable}{lemma}{compressedgeo}\label{lemma:compressed_geometric}
    For any region $\varphi\in \mcA(\mcV)$, any region $\varrho\in \disc{\delta}$ inside $\varphi$, and any point $b\in B$, if $\distance_\delta(\rep\varphi, b)\le 4n$, then $s_\delta(\varrho, b)=\distance_\delta(\rep\varphi,b)\delta$. Furthermore, if $\distance_\delta(\rep\varphi, b)= 4n+1$, then $s_\delta(\varrho, b)\ge (4n+1)\delta$.
\end{restatable}

Next, we show that for each scale $\delta$, the semi-discrete transport plan $\tau_\delta$ and dual weights $(y+\delta\hat{y})(\cdot)$ for the points in $B$ computed by our algorithm at the end of the scale is a $\delta$-optimal transport plan. 

\vspace{0.5em}
\noindent{\bf \texorpdfstring{$\delta$-}-optimality of the computed transport plan.}
Recall that $X_\delta$ denotes the set of representative points of the regions in $\mcA(\mcV)$ and $\hat{\mu}_\delta$ is the discrete distribution over $X_\delta$ computed by our algorithm at step (i).
In the following lemma, we show that any optimal transport plan $\sigma^*$ from $\hat\mu_\delta$ to $\nu$ under distance function $\distance_\delta$ does not transport mass on edges $(\rep\varphi, b)\in X_\delta\times B$ with cost $\distance_\delta(\rep\varphi, b)>4n$.

\begin{restatable}{lemma}{ndeltaslack}\label{lemma:only_2ndelta_slack}
    For any scale $\delta$, let $\sigma^*$ be any optimal transport plan from $\hat\mu_\delta$ to $\nu$. For any point $b\in B$ and any region $\varphi\in \mcA(\mcV)$, if $\sigma^*$ transports mass from $\rep\varphi$ to $b$, then $\distance_\delta(\rep\varphi, b)\le 4n$.
\end{restatable}
\begin{proof}
Let $\tau_{2\delta}, y(\cdot)$ be the $2\delta$-optimal transport plan computed by our algorithm at scale $2\delta$. 
Let $\sigma_{2\delta}$ denote a transformation of $\tau_{2\delta}$ into a discrete transport plan from $\hat{\mu}_\delta$ to $\nu$ by simply setting, for each region $\varphi\in \mcA(\mcV)$, $\sigma_{2\delta}(\rep\varphi, b):=\tau_{2\delta}(\varphi,b)$. Let $\sigma^*$ be any optimal transport plan from $\hat{\mu}_\delta$ to $\nu$, where the cost of each edge $(\rep\varphi, b)$ is set to $\distance_\delta(\rep\varphi, b)$. Define the \emph{residual network} $\mcG$ on the vertex set $X_\delta\cup B$ as follows.
For any pair $(r,b)\in X_\delta\times B$, if $\sigma_{2\delta}(r,b) > \sigma^*(r,b)$, then we add an edge directed from $b$ to $r$ with a capacity $\sigma_{2\delta}(r,b)-\sigma^*(r,b)$; otherwise, if $\sigma_{2\delta}(r,b) < \sigma^*(r,b)$, then we add an edge directed from $r$ to $b$ with a capacity $\sigma^*(r,b) - \sigma_{2\delta}(r,b)$. This completes the construction of the residual network.

For contradiction, suppose there is a pair $(r^*,b^*)\in X_\delta\times B$ such that $\sigma^*(r^*,b^*)>0$ and $\distance_\delta(r^*,b^*)>4n$. From Lemma~\ref{lemma:2delta_slack}, $\sigma_{2\delta}(r^*,b^*)=0$ since $\sigma_{2\delta}$ transports mass only on edges with distance at most $4$. Hence, in the residual network $\mcG$, there is a directed edge from $r^*$ to $b^*$ and by Lemma~\ref{lemma:cycle}, the edge $(r^*,b^*)$ is contained in a simple directed cycle $C=\langle b_1,r_1,\ldots, b_k,r_k\rangle$ in the residual network. Define the cost of the cycle $C$ as 
\[w(C):=\sum_{\langle b,r\rangle \in C}\distance_\delta(r,b) - \sum_{\langle r,b\rangle \in C}\distance_\delta(r,b).\]
Since $\sigma^*$ is an optimal transport plan from $\hat\mu_\delta$ to $\nu$, any cycle $C$ on the residual network have a non-negative cost. Note that the length of $C$ is at most $2n$ since $C$ is a simple
cycle and each point of $B$ appears at most once in $C$. Furthermore, by Lemma~\ref{lemma:2delta_slack}, any directed edge $(b_i,r_i)\in C$ has a distance at most $4$. Finally, by construction, all edges have a non-negative cost. Therefore, 
\begin{align*}
    0 &\le w(C) = \sum_{\langle b,r\rangle \in C}\distance_\delta(r,b) - \sum_{\langle r,b\rangle \in C}\distance_\delta(r,b)\le \sum_{\langle b,r\rangle \in C}4 - \distance_\delta(r^*,b^*)\le 4n - \distance_\delta(r^*, b^*) <0,
\end{align*}
which is a contradiction. Hence, $\sigma^*$ cannot transport mass on edges $(r^*, b^*)$ with cost $\distance_\delta(r^*, b^*)>4n$.
\end{proof}

Let $\sigma_\delta, \hat{y}(\cdot)$ be the optimal transport plan from $\hat\mu_\delta$ to $\nu$ computed at step (ii) of our algorithm, and recall that $\tau_\delta$ is the transport plan from $\mu$ to $\nu$ computed at the end of scale $\delta$. In the following lemma, we show that $\tau_\delta, (y+\delta \hat{y})(\cdot)$ is a $\delta$-optimal transport plan. 

\begin{lemma}\label{lemma:our-delta-optimal}
    For each scale $\delta$, let $(y+\delta\tilde{y})(\cdot)$ denote the set of dual weights for points in $B$ computed at step (iii) of our algorithm. Then, the transport plan $\tau_\delta, (y+\delta\tilde{y})(\cdot)$ is a $\delta$-optimal transport plan.
\end{lemma}
\begin{proof}
    Let $y_\delta(\cdot)$ denote the set of dual weights derived for the representative points of regions in $\disc{\delta}$ using Equation~\eqref{eq:dual_a_assignment} at the beginning of scale $\delta$. Consider a set of dual weights $y'_\delta$ that assigns, for each region $\varrho\in \disc{\delta}$ inside a region $\varphi\in \mcA(\mcV)$, a dual weight $y'_\delta(\rep\varrho):= y_\delta(\rep\varrho) + \hat{y}(\rep\varphi)$. First, we show that the transport plan $\tau_\delta$ along with dual weights $(y+\delta \hat{y})(\cdot)$ and $y'_\delta(\cdot)$ satisfy $\delta$-optimality conditions~\eqref{eq:feasibility_non_matching} and~\eqref{eq:feasibility_matching}. We then show that deriving the dual weights for the representative points of the regions in $\disc{\delta}$ from the dual weights $(y+\delta \tilde{y})(\cdot)$ as in Equation~\eqref{eq:dual_a_assignment} does not violate $\delta$-optimality conditions and conclude that the transport plan $\tau_\delta$ and dual weights $(y+\delta \hat{y})(\cdot)$ for points in $B$ is $\delta$-optimal.

For any region $\varphi\in \mcA(\mcV)$, any region $\varrho\in \disc{\delta}$ inside $\varphi$, and any point $b\in B$, 
\begin{itemize}
\item by Lemma~\ref{lemma:compressed_geometric}, $\distance_\delta(\rep\varphi, b)\delta \le s_\delta(\varrho, b)$. Combining with feasibility condition~\eqref{eq:exact-feasibility_non_matching},
\begin{align*}
    (y+\delta \tilde{y})(b) - y'_\delta(\rep\varrho)&=(y(b)+\delta \tilde{y}(b)) - (y_\delta(\rep\varrho)+\delta \tilde{y}(\rep\varphi))\\ &= (y(b)-y_\delta(\rep\varrho)) + \delta(\tilde{y}(b)-\tilde{y}(\rep\varphi))\nonumber \\& \le (y(b)-y_\delta(\rep\varrho)) +  \distance_\delta(\rep\varphi,b)\delta \le(y(b)-y_\delta(\rep\varrho)) + s_\delta(\varrho,b)\nonumber\\& \le (y(b)-y_\delta(\rep\varrho)) + (\distance(\rep\varrho,b)-y(b)+y_\delta(\rep\varrho)+\delta)\\ &= \distance(\rep\varrho,b)+\delta,
\end{align*}
leading to $\delta$-optimality condition~\ref{eq:feasibility_non_matching}.

\item if $\tau_\delta(\varrho, b)>0$, then $\sigma_\delta$ transports mass from $\rep\varphi$ to $b$, i.e., $\sigma_\delta(\rep\varphi, b)>0$. In this case, by Lemma~\ref{lemma:only_2ndelta_slack}, $\distance_\delta(\rep\varphi, b)\le 4n$ and by Lemma~\ref{lemma:compressed_geometric}, $s_\delta(\varrho,b)=\distance_\delta(\rep\varphi,b)\delta$. Combining with feasibility condition~\eqref{eq:exact-feasibility_matching}, 
\begin{align*}
    (y+\delta \tilde{y})(b) - y'_\delta(\rep\varrho)&=(y(b)+\delta \tilde{y}(b)) - (y_\delta(\rep\varrho)+\delta \tilde{y}(\rep\varphi))\\ &= (y(b)-y_\delta(\rep\varrho)) + \delta(\tilde{y}(b)-\tilde{y}(\rep\varphi))\nonumber \\& = (y(b)-y_\delta(\rep\varrho)) + \distance_\delta(\rep\varphi,b)\delta=(y(b)-y_\delta(\rep\varrho)) + s_\delta(\varrho,b)\\ &\ge (y(b)-y_\delta(\rep\varrho)) + (\distance(\rep\varrho,b)-y(b)+y_\delta(\rep\varrho))\nonumber\\&  = \distance(\rep\varrho,b),
\end{align*}
leading to $\delta$-optimality condition~\ref{eq:feasibility_matching}.
\end{itemize}

In Lemma~\ref{lemma:no_need_for_dual_a} in the appendix, we show that reassigning the dual weights as in Equation~\eqref{eq:dual_a_assignment} does not violate $\delta$-optimality conditions~\eqref{eq:feasibility_non_matching} and~\eqref{eq:feasibility_matching}; hence, $\tau, (y+\delta \hat{y})(\cdot)$ is $\delta$-optimal, as claimed.
\end{proof}

\subsection{Computing Optimal Dual Weights.}\label{sec:optimal_duals} In this section, we show that in addition to computing an $\varepsilon$-close transport cost in the semi-discrete setting, our algorithm can also compute the set of dual weights for the points in $B$ accurately, up to $O(\log \varepsilon^{-1})$ bits. To obtain such accurate set of dual weights, we execute our algorithm for $O(\log(n\Delta/\varepsilon))$ iterations so that the final value of $\delta$ when the algorithm terminates is at most $\varepsilon/5n$. In the following, we show that the dual weight computed for each point in $B$ at the last scale is $\varepsilon$-close to the optimal dual weight value.  

Note that any edge in the graph constructed in Step (i) of our algorithm has a cost at most $4n+1$. Consequently, in Step (ii), the largest dual weight returned by the primal-dual solver is at most $4n+1$\footnote{Any set of dual weights returned by the algorithm can be translated by a fixed value so that the smallest dual weight becomes $0$. Assuming this, it is easy to see that the largest dual weight is $4n+1$.} and in Step (iii), the dual weight of any point $b \in B$ changes by at most $(4n+1)\delta$. Since the dual weight of $b$ becomes the optimal dual weight in the limit, to bound the difference between the current dual weight and the optimal, it suffices if we bound the total change in the dual weights for all scales after scale $\delta\le \varepsilon/5n$. The difference between the optimal dual weight and the current dual weight is at most \[(4n+1)\sum_{i=1}^{\infty} \delta/2^i = (4n+1)\delta\le (4n+1)(\varepsilon/5n) \le \varepsilon.\] Therefore, after $O(\log (n\Delta/\varepsilon))$ iterations of the algorithm, the difference in the optimal dual weight $y(b)$ and the current dual weight of $b$ is at most $\varepsilon$.

\section{Approximation Algorithm for Semi-Discrete Optimal Transport} \label{sec:semi}

In this section, we present our second approximation algorithm for the semi-discrete setting that computes an $\varepsilon$-OT plan in $n \varepsilon^{-O(d)} \text{poly} \log (n)$ expected time. We begin by describing a few notations that help us in presenting our algorithm.
Let $\mu, \nu, A,$ and $B$ be the same as above. For any point $b\in \mbR^d$ and any $r\ge 0$, let $\ball{b}{r}$ denote the Euclidean ball of radius $r$ centered at $b$. Any pair of sets $P, Q \subset \mathbb{R}^d$ is called \emph{$\varepsilon$-well separated} if $\max\{\text{diam}(P), \text{diam}(Q)\} \leq \varepsilon \cdot \min_{(p,q) \in P\times Q} \eucl{p}{q}$. Given a set $S$ of $n$ points in $\mathbb{R}^d$ and a parameter $\varepsilon$, a collection $W = \{(P_1, Q_1), \dots, (P_k, Q_k)\}$ is an \emph{$\varepsilon$-well separated pair decomposition} ($\varepsilon$-WSPD) of $S$ if \emph{(i)} each pair $(P_i, Q_i)$ is $\varepsilon$-well separated, and \emph{(ii)} for any distinct $p,q \in S$, there exists a pair $(P_i, Q_i) \in W$ where $p \in P_i$ and $q \in Q_i$. Given a point set $B \subset \mathbb{R}^d$ and a hypercube $\cell$, we say that $\cell$ is \emph{$\varepsilon$-close to $b \in B$} if $\max_{a \in \cell} \eucl{b}{a} \leq \varepsilon \min_{b' \neq b \in B} \eucl{b'}{b}$. For any parameter $\delta>0$, let $\mbG_\delta$ denote an axis-aligned grid of side-length $\delta$ with a vertex at the origin, i.e., $\mbG_\delta:=[0,\delta]^d + \mbZ^d$. In the remainder of this section, we present our algorithm and analyze its correctness and efficiency.

\subsection{Algorithm.} Here is a brief overview of our algorithm. Let $H$ be a hypercube of side-length $\frac{4}{\varepsilon} \text{diam}(B)$ centered at one of the points of $B$. First, we partition $H$ into a collection of hypercubes such that for each $b \in B$ and all hypercubes $\cell$ except the ones that are $\varepsilon$-close to $b$, the following condition holds: for all $p,q \in \cell$, $\eucl{b}{p} \leq (1+\varepsilon) \eucl{b}{q}$. If a hypercube $\cell$ is $\varepsilon$-close to $b \in B$, then we greedily route the mass of $\mu$ inside $\cell$ to $b$. We then construct a discretization $\hat{\mu}$ of the remaining mass from $\mu$ by collapsing the mass $\mu(\cell)$ of each cell $\cell$ to its center point $c_\cell$. We compute an $\varepsilon$-OT plan $\sigma$ from $\hat{\mu}$ to $\nu$ using the algorithm describe in Section \ref{sec:discrete-OT} and transform $\sigma$ into a semi-discrete transport plan $\tau_\sigma$ by dispersing the mass transportation throughout each hypercube, as described in Section \ref{sec:scaling}. We now describe the algorithm in more detail.

\subsection*{Construction of hypercubes.} Let $W$ denote an $(\frac{\varepsilon}{4})$-WSPD of $B$. For every pair $(B_1, B_2) \in W$, we construct a set of hypercubes closely following the construction of an approximate Voronoi diagram \cite{arya2002linear}, as follows. Let $b_1 \in B_1$ and $b_2 \in B_2$ denote arbitrary representative points of $B_1$ and $B_2$, respectively. For any integer $i = 0, \dots, t = 2\log_2 (2d \varepsilon^{-1})$, define $\delta_i = 2^i \frac{\varepsilon}{2\sqrt{d}} \eucl{b_1}{b_2}$ and let $\partition_i(B_1,B_2)$ denote the set of hypercubes of the grid $\mbG_{\varepsilon \delta_i}$ intersecting $\ball{b_1}{\delta_i} \cup \ball{b_2}{\delta_i}$. For any cell $\cell \in \partition_i(B_1,B_2)$, if there exists a child cell $\cell' \subset \cell$ in $\partition_{i-1}(B_1,B_2)$, then we replace $\cell$ with its $2^d$ child cells to keep all hypercubes interior disjoint. Set $\partition = \bigcup_{(B_1, B_2) \in W} \bigcup_{i=0}^t \partition_i(B_1, B_2)$.

\subsection*{Transporting local mass.} For any point $b \in B$ and some sufficiently small constant $c > 0$, define its local neighborhood to be \[\neighborhood{\varepsilon}{b} = \left\{\cell \in \partition: \, \max_{a \in \cell} \eucl{b}{a} \leq c \,\varepsilon \min_{b' \neq b} \eucl{b'}{b}\right\}.\] For each $b \in B$, we transport the mass locally as follows. If $\nu(b) > 0$ and there exists a hypercube $\cell \subseteq \neighborhood{\varepsilon}{b}$ with $\mu(\cell) > 0$, we transport $\min\{\mu(\cell), \nu(b)\}$ mass from $\cell$ to $b$. If $\mu(\cell) \leq \nu(b)$, we set $\nu(b) = \nu(b) - \mu(\cell)$, delete $\cell$ from $\partition$, and repeat the above step. If $\mu(\cell) > \nu(b)$, we set $\nu(b) = 0$ and scale the mass in $\cell$ down so that $\mu(\cell) = \mu(\cell) - \nu(b)$. This process stops when either $\nu(b) = 0$ or no cell of $\partition$ lies inside $\neighborhood{\varepsilon}{b}$.

\subsection*{Discrete OT on remaining demand.}
Let $\mu'$ and $\nu'$ be the two distributions after transporting the local mass. Note that $\mu'$ and $\nu'$ are not necessarily probability distributions, i.e., the mass of each one of them might not add up to $1$; however, the total mass in $\mu'$ equals that of $\nu'$. Let $\partition$ be the set of remaining hypercubes. Let $\discrete{A} = \{\centerof\cell : \cell \in \partition\} \cup \{c_0\}$ for some $c_0 \in A \setminus H$, where $c_\cell$ denotes the center of $\cell$. Define $\discrete{\mu}(\centerof\cell) = \int_{\cell} \mu'(a) \, da$ for every hypercube $\cell \in \partition$ and let $\discrete{\mu}(c_0) = \int_A \mu'(a) da - \sum_{\cell \in \partition} \discrete{\mu}(\centerof\cell)$. We compute a $(1+\varepsilon)$-approximate discrete transport plan $\sigma$ from $\discrete\mu$ to $\nu'$ using the algorithm described in Section \ref{sec:discrete-OT}. We then convert $\sigma$ into a semi-discrete transport plan $\tau_\sigma$ in a straightforward manner, similar to Section~\ref{sec:scaling}. We return a transport plan $\widetilde{\tau}$ obtained from combining $\tau_\sigma$ with the local mass transportation committed in the previous step in a straight-forward manner. It is easy to confirm that the transport plan $\widetilde{\tau}$ is a transport plan from $\mu$ to $\nu$. This completes the description of our algorithm.

\subsection{Proof of Correctness.} 
In this section, we show that the transport plan computed by our algorithm is a $(1+\varepsilon)$-approximate transport plan from $\mu$ to $\nu$. Recall that as a first step, our algorithm constructs a family $\partition$ of hypercubes. In the following lemma, we enumerate useful properties of these hypercubes.

\begin{lemma} \label{lem:partition-properties}
    For each $\cell \in \partition$ the hypercube $\cell$ satisfies at least one of the following two conditions:
    \begin{enumerate}
        \item For any two points $a_1,a_2 \in \cell$ and any $b \in B$, $\eucl{a_1}{b} \leq (1+\varepsilon) \eucl{a_2}{b}$,

        \item There exists some $b \in B$ such that $\eucl{a}{b} \leq \varepsilon \min_{b' \neq b} \eucl{b'}{b}$ for all $a \in \cell$.
    \end{enumerate}
\end{lemma}

We then use a simple triangle inequality argument similar to \cite{av_scg04} to show that a greedy routing on $\neighborhood{\varepsilon}{b}$ only incurs another $(1+\varepsilon)$-relative error.

\begin{restatable}{lemma}{approxlocalflow}\label{lemma:approx-localflow}
    Let $\tau^*$ be an optimal transport plan between $\mu$ and $\nu$, and let $\widetilde{\tau}$ be the transport plan returned by the algorithm. There exists a transport plan $\hat{\tau}$ such that \textit{(i)} $\hat{\tau} = \widetilde{\tau}$ when restricted to $\bigcup_{b \in B} \neighborhood{\varepsilon}{b}$, and \textit{(ii)} $\plancost(\hat{\tau}) \leq (1 + \varepsilon) \plancost(\tau^*)$.
\end{restatable}

We next show that any mass outside of $H$ can be routed arbitrarily while incurring at most $(1+\varepsilon)$-relative error because any two points $b_1, b_2 \in B$ are approximately equidistant from any $a \in A \setminus H$.

\begin{restatable}{lemma}{approxlargedists}\label{lemma:approx-largedists}
    Let $\widetilde{\tau}$ be the semi-discrete transport plan constructed by our algorithm. Let $\tau$ be any arbitrary transport plan. Then,
    \[\sum_{b \in B} \int_{A \setminus \bigcup_{\cell \in \partition} \cell} \eucl{a}{b} \cdot \widetilde{\tau}(a,b) \; da \leq (1+\varepsilon) \sum_{b \in B} \int_{A \setminus \bigcup_{\cell \in \partition} \cell} \eucl{a}{b} \cdot \tau(a,b) \; da.\]
\end{restatable}

Finally, we consider the mass that lies inside $H$ but does not lie in a cell of $\partition$ that is $\varepsilon$-close to a point of $B$ that has survived. We use the fact that all points within such a cell $\cell$ of $\partition$ are roughly at the same distance from a point of $B$, i.e. for any $p,q \in \cell$ where $\mu'(\cell) > 0$ and for any $b \in B$ where $\nu'(b) > 0$, $\eucl{p}{q} \leq (1+\varepsilon) \eucl{q}{b}$.

\begin{restatable}{lemma}{approxdiscretization}\label{lemma:approx-discretization}
    Let $\hat{\tau}'$ be a transport plan between $\mu'$ and $\nu'$ defined by $\hat{\tau}'(a,b) = \hat{\tau}(a,b)$ if $a \not \in \neighborhood{\varepsilon}{b}$ and $\hat{\tau}'(a,b) = 0$ otherwise. Then $\plancost(\tau_\sigma) \leq (1+\varepsilon) \plancost(\hat{\tau}')$.
\end{restatable}

Lemmas \ref{lemma:approx-localflow}-\ref{lemma:approx-discretization} together imply that our algorithm returns an $\varepsilon$-OT plan.

\begin{lemma}
    Let $\widetilde{\tau}$ be the transport plan computed by our algorithm, and let $\tau^*$ be an optimal transport plan between $\mu$ and $\nu$. Then $\plancost(\widetilde{\tau}) \leq (1+\varepsilon) \plancost(\tau^*)$.
\end{lemma}

\subsection{Efficiency analysis.}\label{sec:semi-continuous-efficiency} Callahan and Kosaraju \cite{callahan1995decomposition} have shown that an $(\frac{\varepsilon}{4})$-WSPD $W$ of $S$ of size $O(n \varepsilon^{-d})$ can be constructed in $O(n(\varepsilon^{-d} + \log n))$ time. For each pair in $W$, our algorithm computes $O(\log \varepsilon^{-1})$ approximate balls, where for each approximate ball, our algorithm adds $O(\varepsilon^{-d})$ hypercubes to $\partition$. Therefore, the collection $\partition$ of hypercubes has size $O(n \varepsilon^{-2d} \log \varepsilon^{-1})$. Hence, partitioning the hypercube $H$ takes $O(n(\log n + \varepsilon^{-2d} \log \varepsilon^{-1}))$ time. Furthermore, computing the mass of $\mu$ inside each hypercube take $O(n\varepsilon^{-2d}\log\varepsilon^{-1}Q)$ time. Finally, note that the discrete OT instance computed by our algorithm has size $O(n\varepsilon^{-2d}\log\varepsilon^{-1})$ and hence, can be solved in $O(n\varepsilon^{-4d-5}\log (n)\log^{2d+5}(\log n)\log(\varepsilon^{-1}))$ time using the algorithm in Section \ref{sec:discrete-OT} when the spread of $B$ is polynomially bounded, leading to Theorem~\ref{theorem:semi-continuous}.

\section{A Near-Linear \texorpdfstring{$\varepsilon$-}-Approximation Algorithm for Discrete OT} \label{sec:discrete-OT}

In this section, we present a randomized Monte-Carlo $(1+\varepsilon)$-approximation algorithm for the discrete OT problem. We now let $\mu, \nu$ be two discrete distributions with support sets $A$ and $B$, respectively, which are finite point sets in $\mathbb{R}^d$. Set $n = |A| + |B|$. We first present an overview of the algorithm, then provide details of the various steps, and finally analyze its correctness and efficiency. Our algorithm can be seen as an adaptation of the boosting framework presented by Zuzic~\cite{zuzic2021simple} to the discrete optimal transport problem; we present an $O(\log \log n)$-approximation algorithm for the discrete OT problem and then boost the accuracy of our algorithm using the multiplicative weights update method and compute a $(1+\varepsilon)$-approximate discrete OT plan. 

\subsection{Overview of the Algorithm.}
At a high level, we compute a hierarchical graph $\spanner = (V, E)$, where $V \supseteq A \cup B$ is a set of points in $\mathbb{R}^d$. The weight of an edge is the Euclidean distance between its endpoints. The construction of $\spanner$ is randomized, and $\spanner$ is a \emph{$(1+\varepsilon)$-spanner} in expectation, i.e., $d_\spanner(a,b)$, the shortest-path distance between $(a,b) \in A \times B$ in $\spanner$ satisfies the condition $\eucl{}{} \leq E[d_\spanner(a,b)] \leq (1+\varepsilon)\eucl{a}{b}$. We formulate the OT problem as a min-cost flow problem in $\spanner$ by setting $\eta(u) = \mu(u)$ if $u \in A$ and $\eta(u) = -\nu(u)$ if $u \in B$. Following a bottom-up greedy approach, we construct a flow $\sigma \colon V \to \mathbb{R}_{\geq 0}$ and dual weights $y: V \to \mathbb{R}$ that satisfy (C1) and (C2) with $\rho = a_1 \log \log n$, where $a_1 > 0$ is a constant:

\begin{description}
    \item[(C1)] $|y(u) - y(v)| \leq \rho \eucl{u}{v}$ \;\; $\forall (u,v)\in E$,
    \smallskip
    \item[(C2)] $\sum_{(u,v) \in E} \sigma(u,v) \eucl{u}{v} \leq \sum_{u \in V} y(u) \eta(u)$.
\end{description}

The first condition guarantees the dual solution $y$ is $\rho$-approximately feasible, while the second condition guarantees that $y$ is non-trivial and the flow $\sigma$ is a $\rho$-approximation. Using such a primal-dual solution, one can use multiplicative-weight-update method (MWU) to boost a $\rho$-approximate flow into a $(1+\varepsilon)$-approximate flow on $\spanner$ by making $O(\rho^2 \varepsilon^{-2} \log n)$ calls to our greedy primal-dual approximation algorithm. We also describe the multiplicative weights procedure in Section \ref{sec:MWU}. Once a $(1+\varepsilon)$-approximate flow is obtained in $\spanner$, then one can simply shortcut paths in $\spanner$ to obtain an $\varepsilon$-OT plan; see e.g. \cite{fox2022deterministic}.

We remark that a $(1+\varepsilon)$-spanner is not needed if only a $O(\log \log n)$-approximation is desired. An $O(\log \log n)$-OT plan can be constructed directly in $O(n \log \log n)$ time using our algorithm. We now describe the details of our algorithm.

\subsection{Constructing a spanner.}\label{subsec:improved_spanner}
We now define the construction of the graph $\spanner$, which is built upon a hierarchical partitioning of $\mathbb{R}^d$ and the tree $\tree$ associated to it.

\subsection*{Hierarchical partitioning.}
\label{subsec:improved_hierarchical_partitioning}
For simplicity, we refer to all $d$-dimensional hypercubes as cells. For any cell $\cell$, let $\ell_\cell$ and $c_\cell$ denote its side-length and center, respectively. Let $\Delta = \frac{\max_{p,q \in A \cup B} \eucl{p}{q}}{\min_{p,q \in A \cup B}\eucl{p}{q}}$ denote the \emph{spread} of $A \cup B$. Additionally, define $\grid(\cell, \ell)$ to be the grid that partitions $\cell$ into new cells of side-length $\ell$. Without loss of generality, assume $A \cup B \subseteq [0,\Delta]^d$.

Let $\cell^*$ be a randomly shifted cell of side-length $2\Delta$ containing all points in $A\cup B$, i.e., $\cell^* = [0,2\Delta]^d - x$ for some $x$ chosen uniformly at random from the hypercube $[0,\Delta]^d$. We construct a hierarchical partition of $\cell^*$ as follows. We designate $\cell^*$ as the root cell of $\tree$. For any cell $\cell$ of $\tree$, define $n_\cell := |(A \cup B) \cap \cell|$ as the number of points of $A \cup B$ contained within $\cell$. We construct $\tree$ recursively as follows. If $n_\cell\le\left(\eps^{-1}\log\log n\right)^{3d}$, $\cell$ is a leaf of $\tree$. Otherwise, using the grid $\grid_\cell = \grid\left(\cell, \ell_\cell/n_\cell^{\frac{1}{3d}}\right)$, we partition $\cell$ into smaller cells of side-length $\ell_\cell/n_\cell^{\frac{1}{3d}}$. We add all non-empty cells of $\grid_\cell$ to $\tree$ as the children of $\cell$ and denote them by $\children\cell$. The height $h$ of $\tree$ is $h=O(\log\log n)$.

For any cell $\cell$ of $\tree$, we define a set of $O((\eps^{-1}dh)^d)$ equal-sized subcells as follows. Define $\subcelldiam{\cell}=\frac{\eps \ell_\cell}{4dh}$ to be the side-length of the subcells of $\cell$. We add all the cells of the grid $\grid(\cell, \subcelldiam{\cell})$ that contain a point of $A \cup B$ as the subcells of $\cell$ and denote the resulting family by $\subcells\cell$.

\subsection*{Vertices and edges of the graph.} The vertex set of $\spanner$ consists of the points $A\cup B$ plus the center point of all non-empty cells and subcells of $\tree$. More precisely,
\[ V = (A\cup B)\cup\bigcup_{\cell\in\tree}\{\centerof\cell\}\cup \{\centerof{\subcell} : \subcell\in\subcells{\cell} \}. \]
The edge set of $\spanner$ consists of two sets of edges per cell of $\tree$.

\begin{enumerate}
    \item If $\cell$ is a non-leaf cell, let 
$\localinstance_\cell = \centerof\cell \cup \left(\bigcup_{\cell'\in\children\cell}\centerof{\cell'}\right) \cup \left(\bigcup_{\subcell\in\subcells\cell}\centerof{\subcell}\right)$
be the set of points composed of the center of $\cell$, centers of its children, and the centers of the subcells of $\cell$. Otherwise, let $\localinstance_\cell = c_\cell \cup ((A \cup B) \cap \cell)$.
We construct a $(1+\varepsilon)$-spanner $\spannercell\cell$ on $\localinstance_\cell$. We add all edges of $\spannercell\cell$ to $\spanner$ and refer to them as \emph{greedy edges}. Note that $|\localinstance_\cell| = |\children\cell| + |\subcells\cell| + 1 = O(n_\cell^{1/3} + (h/\varepsilon)^d)$ for any non-leaf cell.

\item In addition, for any non-leaf cell $\cell$, let $X_\cell = \bigcup_{\cell' \in \children\cell} \bigcup_{\subcell\in\subcells{\cell'}}\centerof\subcell$ be the set of centers of the subcells of the children of $\cell$.  Let $\spannercellp\cell$ be a $(1+\varepsilon)$-spanner constructed on the points in $X_\cell$. We add all the edges of $\spannercellp\cell$ to $\spanner$ and refer to them as \emph{shortcut edges}.
\end{enumerate}

Recall that the weight of every edge in $\spanner$ is the Euclidean distance between its endpoints. The greedy edges are the edges that our greedy algorithm uses to compute a flow, whereas the shortcut edges guarantee that the shortest-path distances in $\spanner$ are a $(1+\eps)$-approximation of the Euclidean distances in expectation. We remark that the shortcut edges are only necessary when applying the MWU method to obtain a $(1+\varepsilon)$-approximate transport plan, otherwise only greedy edges are necessary for a $O(\log \log n)$-approximation.

\begin{figure}
    \centering
    \includegraphics[width=.6\textwidth]{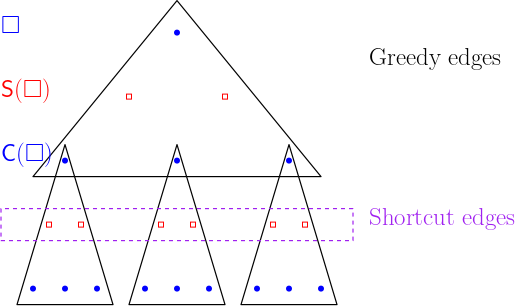}
    \caption{The hierarchical structure of the graph $\spanner$. The vertices of the graph are the centers of the cells (blue disks) and centers of subcells (red squares). For any cell $\cell$, the greedy edges form a spanner on its children and subcells (black triangles) and the shortcut edges form a spanner on the center of the subcells of its children (purple dashed rectangle).}
    \label{fig:improved_spanner}
\end{figure}

For any pair $(a,b) \in A \times B$, let $\pathof{a,b}$ be the shortest path in $\spanner$ from $a$ to $b$ with respect to Euclidean distances along each edge and $\lengthof{\pathof{a,b}}$ to be the cost of $\pathof{a,b}$, i.e. the sum of Euclidean distances of every edge in $\pathof{a,b}$.
The following lemma bounds the size of $\spanner$ and shows that the shortest path metric of $\spanner$, in expectation, $(1+\eps)$-approximates Euclidean distances.

\begin{restatable}{lemma}{discretegraph} \label{lemma:discrete-graph}
    The graph $\spanner$ contains $O(n h)$ vertices and $O(n\varepsilon^{-d} h)$ edges. The max degree of any vertex in $\spanner$ is at most $O(\varepsilon^{-d} \log n)$. Furthermore, for any pair of points $(a,b)$, $\lengthof{\pathof{a,b}} \geq \eucl{a}{b}$ and $\expect{\lengthof{\pathof{a,b}}} \leq (1+3\varepsilon) \eucl{a}{b}$.
\end{restatable}

\subsection{Greedy Primal-Dual Algorithm.}\label{subsec:improved_oracle}

Given the graph $\spanner=(V,E)$ and a demand function $\eta \colon V \to \real$,
we compute a flow $\sigma$ on $\spanner$ satisfying the demand function $\eta$ and a set of dual weights $y$ satisfying the conditions (C1) and (C2) with a parameter $\rho = a_1 \eps^{-1} \log\log n$, where $a_1$ is a constant depending on $d$. It transports as much demand as possible among children of each cell, and routes all excess up the tree $\tree$. Due to the high branching factor of the cells in $\tree$, each subcell contains polynomially many child-subcells. Therefore, subcells cannot simply inherit the dual weights from cells as in \cite{khesin2019preconditioning}, since it might violate condition (C1). Instead, we create a min-cost flow instance for each cell consisting of the centers of its immediate descendants and compute a primal-dual flow on this instance.

\subsection*{Dual assignment and flow function.} We now compute the primal-dual pair $(\sigma,y)$ in a bottom-up manner. At any cell $\cell$ we assume that all excess mass has been routed to $c_{\cell'}$ for each child $\cell' \in \children\cell$, and then route all excess mass from the children of $\cell$ to $c_\cell$. We denote the value of this excess demand in a subtree rooted at $\cell$ as $\bar\eta_\cell$, and it is defined as follows. If $\cell$ is a leaf cell, then $\bar\eta_\cell = \eta(\centerof\cell) + \sum_{p \in (A \cup B) \cap \cell} \eta(p)$. Otherwise, $\bar\eta_\cell = \eta(\centerof\cell) + \sum_{\cell' \in \children\cell} \bar\eta_{\cell'} + \sum_{\subcell \in \subcells{\cell}} \bar\eta_{\subcell}$.

We wish to run Orlin's primal-dual algorithm for min-cost flow on $\spannercell\cell$ \cite{orlin1988faster}. However, we only assume that $\eta$ is a balanced demand function on the whole vertex set of $\spanner$. The total mass in $\spannercell\cell$ defined by $\eta$ may not be balanced on some subgraph $\spannercell\cell$. To resolve this issue, we make $c_\cell$ a sink node that absorbs all excess mass from $\eta$ in the subgraph rooted at $\cell$. We define a local demand function $\eta_\cell:\localinstance_\cell\rightarrow \mathbb{R}$ as follows. For each child $\cell'\in\children\cell$, $\eta_\cell(\centerof{\cell'}) = \bar\eta_{\cell'}$, for each subcell $\subcell \in \subcells\cell$, $\eta_\cell(\centerof\subcell) = \eta(\centerof\subcell)$, and, \[\eta_\cell(\centerof\cell) = -\sum_{\cell' \in \children\cell} \eta_\cell(\centerof{\cell'}) - \sum_{\subcell \in \subcells\cell} \eta_\cell(\centerof{\subcell}).\] Roughly speaking, the demand at the center of a child node $\cell'$ is the surplus/deficit in the subtree rooted at $\cell'$. The demand at the center of $\cell$ is set so that the net excess of demands in $\localinstance_\cell$ is rooted to $\centerof{\cell}$ and similarly, the net deficit of $\localinstance_\cell$ is supplied from $\centerof{\cell}$. The pair $(\spannercell\cell, \eta_\cell)$ is a balanced instance for the min-cost flow. We now run Orlin's primal-dual algorithm for uncapacitated minimum-cost flow to obtain a local primal-dual pair $(\sigma_\cell, y_\cell)$ on $(\spannercell\cell, \eta_\cell)$~\cite{orlin1988faster}. The combination of all flows computed at all cells of $\tree$ satisfies the demand function $\eta$.

Suppose $(\sigma_\cell, y_\cell)$ is the primal-dual flow computed on the local instance $(\spannercell\cell, \eta_\cell), $. For any point $u\in \localinstance_\cell$, we define the dual weight of $u$ as $y(u)\leftarrow y_\cell(u)-y_\cell(\centerof\cell)+y(\centerof\cell)$.
The definition of $y$ synchronizes all the local dual weights computed for each cell of the tree. Additionally, observe that each edge $(u,v)$ of $\spanner$ belongs to a unique local instance $(\spannercell\cell, \eta_\cell)$ of min cost flow. We simply define $\sigma(u,v) = \sigma_\cell(u,v)$, where $\cell$ is the cell for which $(u,v)$ is contained in $\localinstance_\cell$. This completes the construction of our greedy primal-dual algorithm.

\subsection{Multiplicative Weights Update (MWU) Framework.}\label{sec:MWU}
Using one of the known algorithms \cite{charikar2002similarity, indyk2003fast}, we first compute an estimate of the OT cost within a $d \log n$ factor in $O(n \log n)$ time, i.e. we compute a value $\tilde{g}$ such that $\opt \leq \tilde{g} \leq (d\log n) \cdot \opt$. Using this estimate, we perform an exponential search in the range $\left[\frac{\tilde{g}}{d\log n}, \tilde{g}\right]$ with increments of factor $(1+\varepsilon)$. For any guess value $g$, the MWU algorithm either returns a flow $\sigma \colon E \to \real$ with $\plancost(\sigma) \leq (1+\varepsilon) g$ or returns dual weights as a certificate that $g < \opt$. We now describe the MWU algorithm for a fixed value of $g$.

Set $T = 4\rho^2 \varepsilon^{-2} \log |E|$. The algorithm runs in at most $T$ iterations, where in each iteration, it maintains a pre-flow vector $\sigma^t$ satisfying $\plancost(\sigma^t) \leq g$. The pre-flow $\sigma^t$ need not route all demand successfully. Initially, set $\sigma^0(u,v) = \frac{g}{\eucl{u}{v} \cdot |E|}$ for each edge $(u,v)\in E$. For each iteration $t$, define the \emph{residual demand} $\eta_{\text{res}}^t(\cdot)$ as 
\[\eta_{\text{res}}^t(u) = \eta(u) - \sum_{v : (u,v) \in E} (\sigma^{t-1}(u,v) - \sigma^{t-1}(v,u)).\]
Let $(\sigma_{\text{res}}^t, y^t)$ be the primal-dual flow computed by our greedy algorithm for the residual demands $\eta_{\text{res}}^t$. Recall that $(\sigma_{\text{res}}^t, y^t)$ satisfies (C1) and (C2). 
If \ignore{the dual objective cost }$\langle \eta_{\text{res}}^t, y^t \rangle \leq \varepsilon g$, then (C2) implies that $\plancost(\sigma_{\text{res}}^t) \leq \varepsilon g$. Since \ignore{$\greedy$ returns a feasible flow $\sigma_{\text{res}}^t$}$\sigma_{\text{res}}^t$ routes the residual demands, the flow function $\sigma^t = \sigma^{t-1} + \sigma_{\text{res}}^t$ \ignore{is a feasible flow on}routes the original demand $\eta$ with a cost $\plancost(\sigma^t) \leq (1+\varepsilon) g$. In this case, the algorithm returns $\sigma^t$ as the desired flow and terminates.

Otherwise, $\langle \eta_{\text{res}}^t, y^t \rangle > \varepsilon g$ and \ignore{there is some non-trivial large demand which is not yet resolved by $\sigma^{t-1}$. W}we update the flow along each edge $e=(u,v)$ of $G$ based on the slack $\slack^t(u,v) = \frac{y^t(u) - y^t(v)}{\eucl{u}{v}}$ of $e$ with respect to dual weights $y^t$:
\[\sigma^t(u,v) \leftarrow \exp\left(\frac{\varepsilon}{2\rho^2} \slack^t(u,v)\right) \cdot \sigma^{t-1}(u,v).\]

We emphasize that flow along an edge is increasing if the slack is large. Then, one needs to rescale $\sigma^t$ so that its cost is bounded above by $g$. If the algorithm does not terminate within $T$ rounds, we conclude that the value of $g$ is an under-estimate of the cost of the min-cost flow; we increase $g$ by a factor of $(1+\varepsilon)$ and repeat the MWU algorithm. This completes the description of the MWU framework.

\subsection{Analysis.}

The following two lemmas prove that our algorithm satisfies conditions (C1) and (C2) for a sufficiently small approximation factor.

\begin{restatable}{lemma}{improveddualdistortion} \label{lem:loglog-dual-distortion}\label{improveddualdistortion}\label{lemma:improved_dual} 
    For any edge $(u,v) \in E$, $|y(u) - y(v)| \leq O(d^{3/2} h\varepsilon^{-1}) \eucl{u}{v}$.
\end{restatable}

\begin{restatable}{lemma}{strongduality} \label{improvedstrongduality}
    $\sum_{(u,v) \in E} \sigma(u,v) \eucl{u}{v} \leq \sum_{u \in V} y(u) \eta(u)$.
\end{restatable}

Next, we bound the running time of our algorithm. For any cell $\cell$, the algorithm computes an exact primal-dual solution to min-cost flow on $\localinstance_\cell$ with demands $\eta_\cell(\cdot)$ in $O(|\localinstance_\cell|^3)$ time. Each cell $\cell$ satisfies $|\localinstance_\cell| = O\left(n_\cell^{1/3} + (h/\eps)^d\right)$.
The total number of points inside the cells of level $i$ is $n$; i.e, $\sum_{\cell \in \cellsoflevel{i}} n_\cell = n$. Furthermore, the total number of non-empty subcells of the cells at level $i$ is at most $n$; i.e, $\sum_{\cell\in\cellsoflevel{i}}O\left(h/\eps\right)^d \le n$. Therefore, \[\sum_{\cell \in \cellsoflevel{i}} |\localinstance_\cell|^3 = \sum_{\cell \in \cellsoflevel{i}} O\left(n_\cell + \left(h\varepsilon^{-1}\right)^{3d}\right) = O\left(n\left(h\varepsilon^{-1}\right)^{2d}\right).\] Summing over all levels of $\tree$, the total running time of the algorithm is $\tilde{O}\left(n\left(h/\eps\right)^{2d+1}\right)$.

\section*{Acknowledgement}
Work by P.A. and K.Y. has been partially supported by NSF grants IIS-18-14493, CCF-20-07556, and CCF-22-23870. Work by S.R. and P.S. has been partially supported by NSF CCF-1909171 and NSF CCF-2223871. We would like to thank the
anonymous reviewers for their useful comments.

\bibliographystyle{abbrv}
\bibliography{bib.bib}

\newpage
\appendix
\section{Missing Details and Proofs of Section~\ref{sec:scaling_algo_refined}}\label{sec:appendix_scaling}
In this section, we present the missing details and the proofs of the claims made in Section~\ref{sec:scaling_algo_refined}. 

\subsection{Weighted Nearest Neighbor.} Let $w:B\rightarrow \mbR^{\ge 0}$ denote a set of non-negative weights for the points in $B$. Recall that for any pair of points $(a,b)\in A\times B$, the weighted distance of $a$ and $b$ with respect to $w$ is $\distance_w(a,b)=\distance(a,b)-w(b)$. For any point $a\in A$, the \emph{weighted nearest neighbor (WNN)} of $a$ is a point $b\in B$ with the smallest weighted distance to $a$, i.e, a point $b\in B$ satisfying $\distance_w(a,b) = \min_{b'\in B} \distance_w(a,b')$. For any $\delta>0$ and any point $a\in A$, we say that a point $b\in B$ is a \emph{$\delta$-approximate weighted nearest neighbor ($\delta$-WNN)} of $a$ if $\distance_w(a,b)\le \min_{b'\in B} \distance_w(a,b') + \delta$.

\begin{lemma}\label{lemma:delta_WNN}
    Given a transport plan $\tau$ from $\mu$ to $\nu$ and a parameter $\delta>0$, suppose there exists a set of weights $w$ for the points in $B$ such that for any pair of points $(a,b)\in A\times B$ with $\tau(a,b)>0$, the point $b$ is a $\delta$-WNN of $a$ with respect to weights $w$. Then, $\tau$ is a $\delta$-close transport plan from $\mu$ to $\nu$.
\end{lemma}
\begin{proof}
    For any transport plan $\tau'$, we define the weighted cost of $\tau'$, denoted by $\plancost_w(\tau')$, as the cost of the $\tau'$ where the edge costs are replaced with the weighted distance between the points, i.e., $\plancost_w(\tau'):=\sum_{b\in B}\int_A \distance_w(a,b)\tau'(a,b)\, da$. For any transport plan $\tau'$,
    \begin{align}
        \plancost_w(\tau')&=\sum_{b\in B}\int_A \distance_w(a,b)\tau'(a,b)\, da= \sum_{b\in B}\int_A (\distance(a,b)-w(b))\tau'(a,b)\, da\nonumber\\ &= \sum_{b\in B}\int_A \distance(a,b)\tau'(a,b)\, da-\sum_{b\in B} w(b)\int_A \tau'(a,b)\, da\nonumber\\ &= \plancost(\tau') - \sum_{b\in B}w(b)\nu(b).\label{eq:proof_delta_close_1}
    \end{align}
    For any point $a\in A$, suppose $b_a$ denotes any WNN of $a$. Furthermore, for any point $a\in A$, let $\map{\tau}{a}$ denote the set of all points $b\in B$ such that $\tau(a,b) > 0$. Let $\tau^*$ denote any optimal transport plan from $\mu$ to $\nu$. 
    \begin{align}
        \plancost_w(\tau) &= \int_A\sum_{b\in \map{\tau}{a}} \distance_w(a,b)\tau(a,b)\, da\le \int_A\sum_{b\in \map{\tau}{a}} (\distance_w(a,b_a)+\delta)\tau(a,b)\, da \nonumber \\ &= \delta +\int_A \distance_w(a,b_a)\mu(a)\, da \le \delta +\int_A\sum_{b\in B} \distance_w(a,b)\tau^*(a,b)\, da \nonumber\\ &= \plancost_w(\tau^*) + \delta.\label{eq:proof_delta_close_2}
    \end{align}
    Combining Equations~\eqref{eq:proof_delta_close_1} and~\eqref{eq:proof_delta_close_2}, 
    \[\plancost(\tau) = \plancost_w(\tau)+\sum_{b\in B}w(b)\nu(b) \le \plancost_w(\tau) + \delta +\sum_{b\in B}w(b)\nu(b) = \plancost(\tau^*) +\delta,\]
i.e., the transport plan $\tau$ is a $\delta$-close transport plan.
\end{proof}

\subsection{\texorpdfstring{$\delta$-}-Optimal Transport Plan.}\label{sec:feasibility}
Given a continuous distribution $\mu$ defined over a compact bounded set $A$, a discrete distribution $\nu$ defined on a point set $B$, and a parameter $\delta>0$, recall that $\disc{\delta}$ denotes a partitioning over the set $A$, which is the arrangement of all weighted Voronoi diagrams $\vd_w(B)$ for all valid weight vectors $w\in\mbW_\delta$. Recall that for each region $\varrho\in \disc{\delta}$, we refer to its representative point by $\rep{\varrho}$. 
In the following lemma, we show an important property of the partitioning $\disc{\delta}$.

\begin{lemma}\label{lemma:A_delta_prop}
    For any region $\varrho\in \disc{\delta}$, any pair of points $a_1, a_2\in \varrho$, and any valid weight vector $w\in \mbW_\delta$, any $\delta$-WNN of $a_1$ is also a $\delta$-WNN for $a_2$.
\end{lemma}
\begin{proof}
    Suppose a point $b\in B$ is a $\delta$-WNN of the point $a_1$, i.e., for any point $b'\in B$,
    \begin{equation}\label{eq:proof_2d_WNN_1}
        \distance_w(a_1,b) - \delta \le \distance_w(a_1,b').
    \end{equation}
    Define the weights $w_+(\cdot)$ as a set of weights that assigns $w_+(b)=w(b)+\delta$ and $w_+(b')=w(b')$ to each point $b'\neq b$ in $B$. Note that $w_+$ is also a valid weight vector. For any point $b'\neq b$ in $B$, by Equation~\eqref{eq:proof_2d_WNN_1},
    \begin{align}
        \distance_{w_+}(a_1,b) &= \distance(a_1,b)-w_+(b) = \distance(a_1,b) - w(b) - \delta = \distance_w(a_1,b) - \delta \le \distance_w(a_1,b')\nonumber\\ &= \distance_{w_+}(a_1,b'). \label{eq:proof_2d_WNN_2}
    \end{align}
    In other words, $b$ is a WNN for the point $a_1$ with respect to weights $w_+$. Since $w_+\in \mbW_\delta$, by the construction of $\disc{\delta}$, the region $\varrho$ completely lies inside the Voronoi cell of $b$ in the weighted Voronoi diagram $\vd_{w_+}(B)$. As a result, $b$ is also a WNN for the point $a_2$ with respect to the weights $w_+(\cdot)$. Therefore, for any point $b'\neq b$ in B,
    \begin{align*}
        \distance_w(a_2,b) - \delta &= \distance(a_2,b) - w(b) - \delta = \distance_{w_+}(a_2, b) \le \distance_{w_+}(a_2, b')= \distance_{w}(a_2, b'),\label{eq:proof_2d_WNN_3}
    \end{align*}
    i.e., the point $b$ is also a $\delta$-WNN for $a_2$.
\end{proof}

Lemma~\ref{lemma:delta_WNN_main} follows from combining Lemmas~\ref{lemma:delta_WNN} and~\ref{lemma:A_delta_prop} in a straight-forward way.

\begin{restatable}{lemma}{deltaWNNmain}\label{lemma:delta_WNN_main}
    Suppose $\tau$ is a transport plan from $\mu$ to $\nu$ and $w\in \mbW_\delta$ a valid weight vector such that for any pair $(\varrho,b)\in \disc{\delta}\times B$ with $\tau(\varrho,b)>0$, the point $b$ is a $\delta$-WNN of $\rep\varrho$. Then, $\tau$ is a $\delta$-close transport plan from $\mu$ to $\nu$. 
\end{restatable}

In the following lemma, we show that any $\delta$-optimal transport plan $\tau, y(\cdot)$ from $\mu$ to $\nu$ is $\delta$-close.

\deltaoptimal*
\begin{proof}
    To prove this lemma, we first show that 
    for any pair $(\varrho, b)\in \disc{\delta}\times B$ such that $\tau(\varrho, b)>0$, the point $b$ is a $\delta$-WNN of the representative point $\rep\varrho$. Then, by invoking Lemma~\ref{lemma:delta_WNN_main}, we conclude that the transport plan $\tau$ is $\delta$-close, as desired. 

    For any region $\varrho\in \disc{\delta}$ and any point $b\in B$ with $\tau(\varrho, b)>0$, by $\delta$-optimality condition~\eqref{eq:feasibility_matching},
    \begin{equation}\label{eq:proof_delta_optimal_1}
        y(b)-y_\delta(\rep\varrho)\ge \distance(\rep\varrho, b).
    \end{equation}
    Furthermore, for any point $b'\neq b$ in $B$, by $\delta$-optimality condition~\eqref{eq:feasibility_non_matching},
    \begin{equation}\label{eq:proof_delta_optimal_2}
        y(b')-y_\delta(\rep\varrho)\le \distance(\rep\varrho, b') + \delta.
    \end{equation}
    Combining Equations~\eqref{eq:proof_delta_optimal_1} and~\eqref{eq:proof_delta_optimal_2}, 
    \begin{equation*}\label{eq:repX_dist}
        \distance(\rep\varrho,b) - y(b)\le -y_\delta(\rep\varrho) \le \distance(\rep\varrho, b') - y(b') + \delta,
    \end{equation*}
    or equivalently, $\distance_y(\rep\varrho, b) \le \distance_y(\rep\varrho, b') + \delta$, i.e., the point $b$ is a $\delta$-WNN of the representative point $\rep\varrho$. 
\end{proof}

Next, we show that if there exists a transport plan $\tau$ from $\mu$ to $\nu$, a set of dual weight $y(\cdot)$ for points in $B$, and a set of dual weights $y'(\cdot)$ for representative points of the regions in $\disc{\delta}$ that satisfy $\delta$-optimality conditions~\eqref{eq:feasibility_non_matching} and~\eqref{eq:feasibility_matching} (in which $y_\delta(\cdot)$ is replaced with $y'(\cdot)$), then reassigning the dual weights based on Equation~\eqref{eq:dual_a_assignment} does not violate conditions~\eqref{eq:feasibility_non_matching} and~\eqref{eq:feasibility_matching}, i.e., the transport plan $\tau$ and dual weights $y(\cdot)$ for points in $B$ is $\delta$-optimal.
\begin{restatable}{lemma}{noneedfordual}\label{lemma:no_need_for_dual_a}
    For any scale $\delta$, if there exists a transport plan $\tau$ from $\mu$ to $\nu$, a set of dual weights $y(\cdot)$ for points in $B$, and a set of dual weights $y'(\cdot)$ for representative points of regions in $\disc{\delta}$ satisfying $\delta$-optimality conditions~\eqref{eq:feasibility_non_matching} and~\eqref{eq:feasibility_matching}, then $\tau, y(\cdot)$ are $\delta$-optimal.
\end{restatable}
\begin{proof}
    To prove this lemma, we show that conditions~\eqref{eq:feasibility_non_matching} and~\eqref{eq:feasibility_matching} hold when plugging dual weights $y(\cdot)$ for points in $B$ and dual weights $y_\delta(\cdot)$ derived by Equation~\eqref{eq:dual_a_assignment} for representative points of $\disc{\delta}$. For any region $\varrho\in \disc{\delta}$, let $b_\varrho$ denote the weighted nearest neighbor of $\rep\varrho$ in $B$ with respect to weights $y(\cdot)$. 
    For any pair $(\varrho,b)\in \disc{\delta}\times B$, \[y_\delta(\rep\varrho)=y(b_\varrho)-\distance(a,b_\varrho)-\delta \ge y(b)-\distance(\rep\varrho,b)-\delta;\] therefore, the optimality condition~\eqref{eq:feasibility_non_matching} holds for $(\varrho,b)$. Next, we show that the optimality condition~\eqref{eq:feasibility_matching} also holds for all pairs $(\varrho,b)$ with $\tau(\varrho,b)>0$. By condition~\eqref{eq:feasibility_non_matching} on $\tau, y(\cdot), y'(\cdot)$, for any point $b'\in B$, $y'(\rep\varrho)\ge y(b')-\distance(\rep\varrho, b')-\delta$. Therefore,
    \[y(\rep\varrho)\ge \max_{b'\in B}(y(b')-\distance(\rep\varrho,b')-\delta) = y(b_\varrho)-\distance(\rep\varrho,b_\varrho)-\delta = y_\delta(\rep\varrho).\]
    As a result, for the point $b\in B$ with $\tau(\varrho,b)>0$, by condition~\eqref{eq:feasibility_matching} on $\tau, y(\cdot), y'(\cdot)$, we have \[y(b)-y_\delta(\rep\varrho)\ge y(b)-y'(\rep\varrho)\ge \distance(\rep\varrho,b),\]
    and the $\delta$-optimality condition~\eqref{eq:feasibility_matching} holds after replacing $y'(\cdot)$ with $y_\delta(\cdot)$.
\end{proof}

\subsection{Discretizing the Continuous Distribution.}

\compressedgeo*
\begin{proof}
    For any region $\varrho\in \disc{\delta}$, suppose $b_\varrho\in B$ denotes the weighted nearest neighbor of $\rep\varrho$ with respect to weights $y(\cdot)$. For any point $b\in B$, we can rewrite the slack $s_\delta(\varrho,b)$ as follows.
    \begin{align}
        s_\delta(\varrho,b) &= \floor{\frac{\distance(\rep\varrho,b)+\delta-y(b)+y_\delta(\rep\varrho)}{\delta}}\delta\nonumber \\&= \floor{\frac{\distance(\rep\varrho,b)+\delta-y(b)+(y(b_\varrho)-\distance(\rep\varrho,b_\varrho)-\delta)}{\delta}}\delta\nonumber \\&= \floor{\frac{\distance_y(\rep\varrho,b)-\distance_y(\rep\varrho,b_\varrho)}{\delta}}\delta.\label{eq:slack_def_new}
    \end{align}
    For each point $b\in B$, let $V_b=\vor(b)$ denote the weighted Voronoi cell of the point $b$ in the weighted Voronoi diagram $\vd_y(B)$. Recall that for any $i\in [1,4n+1]$, $V_b^i$ denotes the $i$-expansion of the weighted Voronoi cell of the point $b$. For any pair $(\varrho,b)\in \disc{\delta}\times B$, if $\rep\varrho$ lies inside $V_b$, then $b$ is the WNN of $\rep\varrho$ and by Equation~\eqref{eq:slack_def_new}, $s_\delta(\varrho,b)=0$. Otherwise, suppose the point $\rep\varrho$ lies inside $V_b^i$ for some $i\in[1,4n+1]$. Let $y'(\cdot)$ denote a set of dual weights for the point set $B$ that assigns $y'(b)=y(b)+i\delta$ to the point $b$ and $y'(b')=y(b')$ to each point $b'\neq b$ in $B$. Since $\rep\varrho$ lies inside $V_b^i$, then $b$ is the weighted nearest neighbor of $\rep\varrho$ with respect to weights $y'(\cdot)$, i.e., $\distance_{y'}(\rep\varrho,b)<\distance_{y'}(\rep\varrho,b')$ for each point $b'\neq b\in B$. Therefore, 
    \begin{align*}
        \distance_y(\rep\varrho,b) &= \distance(\rep\varrho,b)-y(b) = \distance(\rep\varrho,b)-(y'(b)-i\delta) = \distance_{y'}(\rep\varrho,b)+i\delta \\ &< \distance_{y'}(\rep\varrho,b_\varrho)+i\delta = \distance_y(\rep\varrho,b_\varrho)+i\delta.
    \end{align*}
    Plugging into Equation~\eqref{eq:slack_def_new}, $s_\delta(\varrho,b) = \floor{\frac{\distance_y(\rep\varrho,b)-\distance_y(\rep\varrho,b_\varrho)}{\delta}}\delta<i\delta$ for any region $\varrho\in \disc{\delta}$ inside $V_b^i$. Furthermore, for any region $\varrho\in \disc{\delta}$ outside of $V_b^i$, the WNN of $\varrho$ with respect to weights $y'(\cdot)$ remains to be $b_\varrho$ and we have $\distance_{y'}(\rep\varrho,b) > \distance_{y'}(\rep\varrho,b_\varrho)$. Therefore, 
    \begin{equation*}
        \distance_y(\rep\varrho,b) = \distance_{y'}(\rep\varrho,b)+i\delta > \distance_{y'}(\rep\varrho,b_\varrho)+i\delta = \distance_y(\rep\varrho,b_\varrho)+i\delta.
    \end{equation*}
    Plugging into Equation~\eqref{eq:slack_def_new}, $s_\delta(\varrho,b) = \floor{\frac{\distance_y(\rep\varrho,b)-\distance_y(\rep\varrho,b_\varrho)}{\delta}}\delta\ge i\delta$ for any region $\varrho\in \disc{\delta}$ outside $V_b^i$. Thus, for any point $b\in B$, any region $\varphi\in \mcA(\mcV)$, and any $\varrho\in \disc{\delta}$ inside $\varphi$,
    \begin{itemize}
        \item[--] if $\rep\varrho$ lies inside $V_b^1$, then $s_\delta(\varrho,b)=0$. In this case, $\rep\varphi$ also lies inside $V_b^1$ and $\distance_\delta(\rep\varphi,b)=0$,
        \item[--] if $\rep\varrho$ lies inside $V_b^{i+1}\setminus V_b^i$ for some $i\in [1,4n]$, then $s_\delta(\varrho,b)=i\delta$. In this case, $\rep\varphi$ also lies in $V_b^{i+1}\setminus V_b^i$ and $\distance_\delta(\rep\varphi,b)=i$, and
        \item[--] if $\rep\varrho$ lies outside $V_b^{4n+1}$, then $s_\delta(\varrho,b)\ge (4n+1)\delta$. In this case, $\rep\varphi$ also lies outside of $V_b^{4n+1}$ and $\distance_\delta(\rep\varphi,b)=4n+1$.
    \end{itemize}
    This completes the proof of this lemma.
\end{proof}

\subsection{\texorpdfstring{$\delta$-}-Optimality of the Computed Transport Plan.}

\begin{lemma}\label{lemma:2delta_slack}
    Let $\tau_{2\delta}, y(\cdot)$ be any $2\delta$-optimal transport plan from $\mu$ to $\nu$, where the dual weights of points in $B$ are integer multiples of $2\delta$. Then, for any region $\varrho\in \disc{\delta}$ and any point $b\in B$, if $\tau_{2\delta}(\varrho, b)>0$, then $s_\delta(\varrho,b)\le 4\delta$.
\end{lemma}
\begin{proof}
    Let $\varrho^*$ denote the region in $\disc{2\delta}$ containing $\varrho$ (by construction, it can be easily confirmed that the set of valid weight vectors $\mbX_{2\delta}$ is a subset of $\mbW_\delta$ and hence, each region in $\disc{\delta}$ completely lies inside a region in $\disc{2\delta}$).
    Define $b_\varrho$ to be the weighted nearest neighbor of $\rep{\varrho}$ (and consequently $\rep{\varrho^*}$) with respect to weights $y(\cdot)$. By Equation~\eqref{eq:dual_a_assignment}, $y_{2\delta}(\rep{\varrho^*})=y(b_\varrho)-\distance(\rep{\varrho^*}, b_\varrho)-2\delta$ and $y_\delta(\rep{\varrho}) = y(b_\varrho)-\distance(\rep{\varrho},b_\varrho)-\delta$. Hence,
    \begin{align}
        s_\delta(\varrho, b) &= \floor{\frac{\distance(\rep{\varrho}, b)+\delta-y(b)+y_\delta(\rep{\varrho})}{\delta}}\delta\nonumber\\ &= \floor{\frac{\distance(\rep{\varrho}, b)+\delta-y(b)+(y(b_\varrho)-\distance(\rep{\varrho},b_\varrho)-\delta)}{\delta}}\delta\nonumber\\ &= \floor{\frac{\distance_y(\rep{\varrho}, b)-\distance_y(\rep{\varrho},b_\varphi)}{\delta}}\delta\le \floor{\frac{\distance_y(\rep{\varrho^*}, b)-\distance_y(\rep{\varrho^*}, b_{\varrho})}{\delta}}\delta + 2\delta,\label{eq:proof2deltaslack2}
    \end{align}
    where the last inequality is resulted from Lemma~\ref{lemma:distance_difference} below.
    Finally, from the $2\delta$-optimality condition~\eqref{eq:feasibility_matching} on $\tau_{2\delta}, y(\cdot)$,  
    \begin{equation*}
        \label{eq:proof2deltaslack1}
        \distance(\rep{\varrho^*}, b) \le  y(b)-y_{2\delta}(\rep{\varrho^*}) = y(b) - (y(b_\varrho)-\distance(\rep{\varrho^*}, b_\varrho)-2\delta) .
    \end{equation*}
    Hence,
    \begin{equation}
        \label{eq:proof2deltaslack4}
        \distance_y(\rep{\varrho^*}, b) - \distance_y(\rep{\varrho^*}, b_\varrho)\le 2\delta.
    \end{equation}
    Plugging Equations~\eqref{eq:proof2deltaslack4} into Equation~\eqref{eq:proof2deltaslack2},
    \[s_\delta(\varrho,b)\le \floor{\frac{\distance_y(\rep{\varrho^*}, b)-\distance_y(\rep{\varrho^*},b_\varrho)}{\delta}}\delta+2\delta \le 4\delta, \]
    as claimed.
\end{proof}

\begin{lemma}\label{lemma:distance_difference}
    For any region $\varrho^*\in \disc{2\delta}$, any pair of points $a_1, a_2\in \varrho^*$, and any pair of points $b_1, b_2\in B$, $\floor{\frac{\distance(a_1,b_1)-\distance(a_1,b_2)}{\delta}}\delta \le \floor{\frac{\distance(a_2,b_1)-\distance(a_2,b_2)}{\delta}}\delta + 2\delta$.
\end{lemma}
\begin{proof}
    To prove this lemma, we first construct a valid weight vector $w\in\mbW_{2\delta}$ such that in the Voronoi diagram $\vd_w(B)$, the region $\varrho^*$ lies inside the Voronoi cell of $b_1$, which gives us $\distance_w(a_1,b_1)\le \distance_w(a_1,b_2)$. Then, we increase the weight of $b_2$ in $w$ by $2\delta$ and obtain another valid weight vector $w_+\in \mbW_{2\delta}$ such that $\varrho^*$ now lies inside the Voronoi cell of $b_2$ and conclude $\distance_w(a_2,b_2)\le \distance_w(a_2,b_1)+2\delta$. Combining the two bounds, we get $\distance(a_1,b_1)-\distance(a_1,b_2)\le \distance(a_2,b_1)-\distance(a_2,b_2) + 2\delta$, leading to the lemma statement. We describe the details below.

    Consider a valid weight vector $w\in\mbW_{2\delta}$ that assigns $w(b_1)=\ceil{\frac{\distance(\rep{\varrho^*},b_1)}{2\delta}}2\delta, w(b_2)=\ceil{\frac{\distance(\rep{\varrho^*},b_2)}{2\delta}}2\delta$, and $w(b')=0$ for each $b'\neq b_1,b_2$ in $B$. Without loss of generality, assume $\distance_w(\rep{\varrho^*},b_1)< \distance_w(\rep{\varrho^*},b_2)$\footnote{If $\distance_w(\rep{\varrho^*},b_1)\ge \distance_w(\rep{\varrho^*},b_2)$, one can simply decrease $w(b_2)$ by $2\delta$ and follow a very similar argument.}. By construction,
    \begin{equation*}\label{eq:distance_difference1}
        -2\delta < \distance_w(\rep{\varrho^*},b_1)< \distance_w(\rep{\varrho^*},b_2)\le 0\le \min_{b'\in B, b'\neq b_1,b_2}\distance_w(\rep{\varrho^*},b).
    \end{equation*}
    Hence, the point $\rep{\varrho^*}$ and consequently the region $\varrho^*$ lie inside the Voronoi cell of $b_1$ in $\vd_w(B)$. Therefore,
    \begin{align}
        \distance(a_1,b_1)-w(b_1) = \distance_w(a_1,b_1)&\le \distance_w(a_1,b_2)=\distance(a_1,b_2)-w(b_2).
    \nonumber\\\label{eq:distance_difference3}
        \distance(a_1,b_1)-\distance(a_1,b_2)&\le w(b_1)-w(b_2).
    \end{align}
    Next, consider the weight vector $w_+\in \mbW_{2\delta}$ that assigns $w_+(b_2)=w(b_2)+2\delta$ and $w_+(b')=w(b')$ for all points $b'\neq b_2$ in $B$. In this case, 
    \begin{equation*}
        \distance_{w_+}(\rep{\varrho^*},b_2)\le -2\delta < \distance_{w_+}(\rep{\varrho^*},b_1)=\distance_{w}(\rep{\varrho^*},b_1)\le 0\le \min_{b'\in B, b'\neq b_1,b_2}\distance_w(\rep{\varrho^*},b).
    \end{equation*}
    Therefore, the point $\rep{\varrho^*}$ and consequently the region $\varrho^*$ lie inside the Voronoi cell of $b_2$ in $\vd_{w_+}(B)$. Therefore,
    \begin{align}
        \distance(a_2,b_2)-(w(b_2)+2\delta)= \distance_{w_+}&(a_2,b_2)\le \distance_{w_+}(a_2,b_1)=\distance(a_2,b_1) - w(b_1),\nonumber\\
        \label{eq:distance_difference4}
        w(b_1)-w(b_2)&\le \distance(a_2,b_1)-\distance(a_2,b_2) + 2\delta.
    \end{align}
    Combining Equations~\eqref{eq:distance_difference3} and~\eqref{eq:distance_difference4},
    \begin{align*}
        \distance(a_1,b_1)-\distance(a_1,b_2)\le w(b_1)&-w(b_2)\le \distance(a_2,b_1)-\distance(a_2,b_2) + 2\delta,\nonumber\\
        \floor{\frac{\distance(a_1,b_1)-\distance(a_1,b_2)}{\delta}}\delta &\le \floor{\frac{\distance(a_2,b_1)-\distance(a_2,b_2)}{\delta}}\delta + 2\delta.
    \end{align*}
\end{proof}

\subsubsection*{Residual Network.}
Given two transport plans $\sigma_1$ and $\sigma_2$ from $\hat{\mu}_\delta$ to $\nu$, we define the residual network $\mcG(\sigma_1, \sigma_2)$ on the vertex set $X_\delta\cup B$ as follows.
Define $\sigma:= \sigma_1 - \sigma_{2}$ to be a function that assigns, for any pair $(r,b)\in X_\delta\times B$, $\sigma(r,b)= \sigma_1(r,b) - \sigma_{2}(r,b)$. For any pair $(r,b)\in X_\delta\times B$, if $\sigma(r,b) > 0$, then we add an edge directed from $b$ to $r$ with a capacity $\sigma(r,b)$; otherwise, if $\sigma(r,b) < 0$, then we add an edge directed from $r$ to $b$ with a capacity $|\sigma(r,b)|$.

\begin{lemma}\label{lemma:cycle}
    Given any two transport plans $\sigma_1$ and $\sigma_2$ from $\hat{\mu}_\delta$ to $\nu$, for any directed edge $(r,b)\in X_\delta\times B$ in the residual network $\mcG(\sigma_1, \sigma_2)$, there exists a directed cycle $C$ in $\mcG(\sigma_1, \sigma_2)$ that contains the edge $(r,b)$.
\end{lemma}
\begin{proof}
    To prove this lemma, we conduct a DFS-style search from the point $b$ in the residual network to compute a directed path $P$ from $b$ to $r$. This proves the lemma since concatenating the edge $(r,b)$ to $P$ results in a directed cycle on the residual network containing $(r,b)$.
    Our proof relies on the following observation: Since both $\sigma_1$ and $\sigma_2$ are transport plans from $\hat{\mu}_\delta$ to $\nu$, by the construction of the residual network, for any point $u\in X_\delta\cup B$, the total capacity of incoming edges to $u$ is equal to the total capacity of outgoing edges from $u$. 

    We conduct a DFS-style procedure that grows a path $P=\langle r=p_0, b=p_1, p_2, \ldots, p_k\rangle$ as follows. Initially, we set $P=\langle r=p_0, b=p_1\rangle$. At each step, for the last point $p_k$ of the path $P$, let $N(p_k)$ denote the set of all outgoing edges from $p_k$. Note that since there exists an incoming edge $(p_{k-1}, p_k)$ in the residual graph, by the observation stated above, $N(p_k)$ is not empty. Consider any point $p\in N(p_k)$.
    \begin{itemize}
        \item If $p=r$, then $P\circ (p_k,r)$ is a cycle containing $(r,b)$, as desired.
        \item Otherwise, if $p$ already exists in the path $P$ as $p_i$ for some $i\ge 1$, then we have found a cycle $C=\langle p_i, p_{i+1},\ldots, p_k,p_i\rangle$. We ``cancel'' this cycle as follows. Define the capacity of the cycle $c(C)$ as the minimum capacity of all edges on $C$. We then decrease the capacity of all edges on $C$ by $c(C)$ and for those that now have a zero capacity, we simply remove them from the residual network. We set $P=\langle p_0=r, p_1=b, \ldots, p_i\rangle$ and continue our search.
        \item Otherwise, we add $p$ as $p_{k+1}$ to the path $P$ and continue the search.
    \end{itemize}
    Note that since all edges in the residual network have a positive finite capacity at all times, each time we cancel a cycle reduces the total capacity of the edges of the residual network. Furthermore, the length of the path $P$ will never be more than $2n$, as there are only $n$ points in the set $B$ and the residual network is a bipartite graph. Hence, our DFS-style procedure will terminate by returning a cycle containing $(r,b)$.
\end{proof}

\section{Missing Details of Section \ref{sec:semi}} \label{sec:semi-appendix}

Without loss of generality, we will assume that $\varepsilon < \frac{1}{2}$. Otherwise, we can replace $\varepsilon$ with $\min\{\frac{1}{3}, \varepsilon\}$ without increasing runtime dependence on $\varepsilon$. This choice of $\varepsilon$ is important to guarantee that all neighborhoods $\neighborhood{\varepsilon}{b}$ are disjoint.

The following lemmas roughly split the edge costs into three cases. First, we observe that we can safely transport mass to any point $b\in B$ from regions of $A$ that have extremely small distances to $b$. Second, we observe that we can arbitrarily transport any mass of $\mu$ that is far enough from all points of $B$ at the cost of an small error. Finally, given a box $\cell$ with the property that for any point $b\in B$, the point $b$ is $(1+\varepsilon)$-approximately equidistant from all points inside $\cell$, the mass of $\mu$ inside $\cell$ can be moved to the center of $\cell$ without too much sacrifice.

\approxlocalflow*

\begin{proof}
    We break the proof of this lemma into two stages. For stage I, we argue that there exists an intermediary transport plan $\hat{\tau}_1$ between $\mu$ and $\nu$ where \textit{(i)} $\int_{\neighborhood{\varepsilon}{b}}\hat{\tau}_1(a,b) \; da$ is as large as possible for all $b$ and \textit{(ii)} $\plancost(\hat{\tau}_1) \leq (1+O(\varepsilon)) \plancost(\tau^*)$. For stage II, we argue that from $\hat{\tau}_1$, the choice of mass within each neighborhood $\neighborhood{\varepsilon}{b}$ which is greedily coupled with $b$ can be swapped so that $\hat{\tau}$ agrees with $\tau$ while only incurring a $(1+\varepsilon)$ approximation error.

    \textbf{Stage I:}
    Let $b$ be an arbitrary element of $B$. Suppose $\int_{\neighborhood{\varepsilon}{b}} \tau^*(x, b) \, dx < \min \left\{\int_{\neighborhood{\varepsilon}{b}} \mu(x) \, dx, \nu(b)\right\}$, i.e. there is some mass within the approximate ball $\neighborhood{\varepsilon}{b}$ which could be routed to $b$ by $\tau^*$ but is instead routed to some point $b' \neq b$. Then there exist sets $U \subseteq \neighborhood{\varepsilon}{b}, V \subseteq A \setminus \neighborhood{\varepsilon}{b}$ and $B' \subseteq B \setminus \{b\}$ where $\sum_{b' \in B'} \int_U \tau^*(u, b') \, du = \int_V \tau^*(v, b) \, dv > 0$ (some mass in $\neighborhood{\varepsilon}{b}$ is routed away from $b$ and an equal mass outside $\neighborhood{\varepsilon}{b}$ is routed to $b$ via $\tau^*$).
    
    Since $U \subseteq \neighborhood{\varepsilon}{b}$, where the radius of the approximate ball $\neighborhood{\varepsilon}{b}$ is $\varepsilon$ and $\eucl{b}{b'} \geq 1$ for all $b' \in B'$, we know that
    \[\eucl{u}{b'} \geq (1-\varepsilon)\eucl{b'}{b} \text{ and } \eucl{u}{b} \leq \varepsilon \eucl{b'}{b}\]
    for all $u \in U$ and $b' \in B'$. Hence, $\eucl{u}{b} \leq \frac{\varepsilon}{1-\varepsilon} \eucl{u}{b'}$. Moreover, by the triangle inequality we can conclude that for any $u,v \in U \times V$ and $b' \in B'$,
    \[\eucl{v}{b'} \leq \eucl{v}{b} + \eucl{b}{u} + \eucl{u}{b'} \leq \eucl{v}{b} + \left(1+\frac{\varepsilon}{1-\varepsilon}\right) \eucl{u}{b'}.\]
    Let $\tau_{1}$ be defined as the transport plan which routes the mass of $U$ to $b$, the mass of $V$ to $b'$, and equals $\tau^*$ elsewhere. It follows that
    \begin{align*}
        &\int_{U} \eucl{u}{b} \hat{\tau}_1(u,b) \, du + \sum_{b' \in B'} \int_{V} \eucl{v}{b'} \hat{\tau}_1(v, b') \, dv \\
        &\leq \left(1+\frac{2\varepsilon}{1-\varepsilon}\right)\sum_{b' \in B'} \int_{U} \eucl{u}{b'} \tau^*(u,b') \, du + \int_{V} \eucl{v}{b} \tau^*(v, b) \, dv.
    \end{align*}
    Since $\varepsilon \in (0,\frac{1}{2})$, we note that $\frac{2\varepsilon}{1-\varepsilon} \leq 4\varepsilon$. Furthermore, since $\hat{\tau}_1 = \tau^*$ outside of $(U \cup V) \times B' \cup \{b\}$, we deduce $\plancost(\hat{\tau}_1) \leq (1+4\varepsilon)\plancost(\tau^*)$.

    Finally, we note that $\neighborhood{\varepsilon}{b} \cap \neighborhood{\varepsilon}{b'} = \varnothing$ for all $b \neq b'$ since $\varepsilon \in (0,\frac{1}{2})$. Therefore, the $(1+4\varepsilon)$ approximation factor is incurred at most once for each $U \subseteq \neighborhood{\varepsilon}{b}$ which is rerouted. Since each operation reroutes a maximal amount of mass in $\neighborhood{\varepsilon}{b}$ to $b$ and every pair of neighborhoods is disjoint, we conclude that after $n$ such $U,V$ swaps a satisfactory transport plan $\hat{\tau}_1$ has been constructed from $\tau^*$.

    \textbf{Stage II:}
Let $b$ be an arbitrary element of $B$, and let $\neighborhood{\varepsilon}{b}$ be the approximate ball of radius $\varepsilon$ centered at $b$. Suppose there exist disjoint sets $X, Y \subset \neighborhood{\varepsilon}{b}$ and some $B' \subset B \setminus \{b\}$ where
\[\int_X \tau(x, b) \, dx = \int_Y \hat{\tau}_1(y, b) \, dy = \sum_{b' \in B'} \int_X \hat{\tau}_1(x, b') \, dx > 0.\]
In the same manner as stage I, we now show that for any $x,y \in X \times Y$ and $b' \in B'$,
\[\eucl{x}{b} + \eucl{y}{b'} \leq (1+6\varepsilon)\eucl{x}{b'} + \eucl{y}{b}.\]
Let $x \in X, y \in Y$ and $b' \in B'$ be arbitrarily chosen. First, observe that
\[\eucl{y}{b'} \leq \eucl{x}{b'} + \eucl{x}{b} + \eucl{y}{b} \leq \eucl{x}{b'} + 2\varepsilon\]
by triangle inequality and condition 2 of Lemma \ref{lem:partition-properties}. Additionally, since $x \in X \subseteq \neighborhood{\varepsilon}{b}$, we note that $\eucl{x}{b} \leq \varepsilon$. Now we can use the triangle inequality to bound
\[\eucl{x}{b'} \geq \eucl{b}{b'} - \eucl{x}{b} \geq 1 - \alpha.\]
Combining these inequalities and $\varepsilon \in (0, \frac{1}{2})$, we conclude
\begin{align*}
    \eucl{y}{b'} + \eucl{x}{b} &\leq \eucl{x}{b'} + 3\varepsilon \\
    &\leq (1 + \frac{3\varepsilon}{1-\varepsilon}) \eucl{x}{b'} + \eucl{y}{b} \\
    &\leq (1 + 6\varepsilon)\eucl{x}{b'} + \eucl{y}{b}.
\end{align*}
Let $\hat{\tau}$ be defined as the transport plan which routes the mass of $Y$ to $B'$, the mass of $X$ to $b$, and equals $\hat{\tau}_1$ elsewhere. It follows that
\begin{align*}
    &\int_{X} \eucl{x}{b} \hat{\tau}(x,b) \, dx + \sum_{b' \in B'} \int_{Y} \eucl{y}{b'} \hat{\tau}(y, b') \, dy \\
    &\leq \left(1+6\varepsilon\right)\sum_{b' \in B'} \int_{X} \eucl{x}{b'} \hat{\tau}_1(x,b') \, dx + \int_{Y} \eucl{y}{b} \hat{\tau}_1(y, b) \, dv.
\end{align*}
Furthermore, since $\hat{\tau}_1 = \hat{\tau}$ outside of $(X \cup Y) \times (B' \cup \{b\})$, we deduce $\plancost(\hat{\tau}_1) \leq (1+6\varepsilon)\plancost(\tau^*)$. Repeat for every $b \in B$ and every $X,Y \subseteq \neighborhood{\varepsilon}{b}$ and by construction we then have $\hat{\tau} = \tau$ on $\neighborhood{\varepsilon}{b} \times \{b\}$ for all $b \in B$. Note that no neighborhoods $\neighborhood{\varepsilon}{b}$ intersect, so the cost approximation factor does not increase since no set $X$ is swapped more than once. We conclude that $\plancost(\hat{\tau}) \leq (1+22\varepsilon) \plancost(\tau^*)$.
\end{proof}

\approxlargedists*

\begin{proof}
    Suppose $a \in A \setminus \bigcup_{\cell \in \partition} \cell$. Then $a$ satisfies $\eucl{a}{b_1} \geq \frac{1}{\varepsilon} \eucl{b_1}{b_2}$ for all $b_1, b_2 \in B$. By the triangle inequality, we note that $\eucl{a}{b_2} \leq (1+\varepsilon)\eucl{a}{b_1}$ for all such $a \in A$ and $b_1,b_2 \in B$. 

    Since $\tau$ and $\widetilde{\tau}$ are both feasible transport plans, they satisfy $\sum_{b \in B} \tau(a,b) = \sum_{b \in B} \widetilde{\tau}(a,b)$. We deduce that
    \[\sum_{b \in B} \eucl{a}{b} \cdot \tau(a,b) \leq (1+\varepsilon) \sum_{b \in B} \eucl{a}{b} \cdot \widetilde{\tau}(a,b)\]
    by combining the previous two statements. Finally, integrating over all $a \in A \setminus \bigcup_{\cell \in \partition} \cell$ gives us the desired result

    \[\sum_{b \in B} \int_{A \setminus \bigcup_{\cell \in \partition} \cell} \eucl{a}{b} \cdot \tau(a,b) \; da \leq (1+\varepsilon) \sum_{b \in B} \int_{A \setminus \bigcup_{\cell \in \partition} \cell} \eucl{a}{b} \cdot \widetilde{\tau}(a,b) \; da.\]
\end{proof}

\approxdiscretization*

\begin{proof}
Note that
\[\plancost(\tau) = \sum_{b \in B} \int_{A \setminus \neighborhood{\varepsilon}{b}} \eucl{a}{b} \tau(a,b) \, da + \sum_{b \in B} \int_{\neighborhood{\varepsilon}{b}} \eucl{a}{b} \tau(a, b) \, da\]
where $\neighborhood{\varepsilon}{b}$ again denotes the approximate ball centered at $b$ of radius $\varepsilon$. We can analogously claim
\begin{align*}
    \plancost(\hat{\tau}) &= \sum_{b \in B} \int_{A \setminus \neighborhood{\varepsilon}{b}} \eucl{a}{b} \hat{\tau}(a,b) \, da + \sum_{b \in B} \int_{\neighborhood{\varepsilon}{b}} \eucl{a}{b} \hat{\tau}(a, b) \, da \\
    &= \sum_{b \in B} \int_{A \setminus \neighborhood{\varepsilon}{b}} \eucl{a}{b} \hat{\tau}(a,b) \, da + \sum_{b \in B} \int_{\neighborhood{\varepsilon}{b}} \eucl{a}{b} \tau(a, b) \, da,
\end{align*}
where the second equality follows from the fact that $\hat{\tau} = \tau$ on $\cup_{b \in B} (\neighborhood{\varepsilon}{b} \times b)$. It therefore suffices to compare the transport plans on the pairs $a,b \in A \times B$ of (approximate) distance greater than $\varepsilon$. That is,
\[\plancost(\tau) - \plancost(\hat{\tau}) = \sum_{b \in B} \int_{A \setminus \neighborhood{\varepsilon}{b}} \eucl{a}{b} (\tau(a,b) - \hat{\tau}(a,b)) \, da.\]
For simplicity, let $\mathcal{Z} = \bigcup_{\cell \in \partition} \cell$, $\mathcal{X} = A \setminus \mathcal{Z}$, and define $\mathcal{Z}_b := \mathcal{Z} \setminus \neighborhood{\varepsilon}{b}$ for each $b \in B$. Then, we observe
\[\plancost(\tau) - \plancost(\hat{\tau}) = \sum_{b \in B} \left[\int_{\mathcal{X}} \eucl{a}{b} (\tau(a,b) - \hat{\tau}(a,b)) \, da 
 + \int_{\mathcal{Z}_b} \eucl{a}{b} (\tau(a,b) - \hat{\tau}(a,b)) \, da\right].\]
 For convenience, define $\tau' = \tau - \hat{\tau}$ and let $\partition_b = \{\cell \in \partition : \cell \subseteq \neighborhood{\varepsilon}{b}\}$. Additionally define the discrete plans $\hat{\sigma}$ and $\sigma'$ by $\hat{\sigma}(c_\cell, b) := \int_\cell \hat{\tau}(a,b) \; da$ and $\sigma'(c_\cell, b) = \sigma(c_\cell, b) - \hat{\sigma}(c_\cell, b)$. We conclude that
 \begin{align*}
     \plancost(\tau) - \plancost(\hat{\tau}) &= \sum_{b \in B} \left[\int_{\mathcal{X}} \eucl{a}{b} \tau'(a,b) \, da 
 + \int_{\mathcal{Z}_b} \eucl{a}{b} \tau'(a,b) \, da\right] \\
 &\leq \sum_{b \in B} \left[\int_{\mathcal{X}} \eucl{a}{b} \tau'(a,b) \, da 
 + (1+\varepsilon) \sum_{\cell \in \partition \setminus \partition_b} \eucl{\centerof\cell}{b} \sigma'(\centerof\cell,b)\right] \\
 &\leq \varepsilon \sum_{b \in B} \left[\int_{\mathcal{X}} \eucl{a}{b} \hat{\tau}(a,b) \, da 
 + (1+\varepsilon) \sum_{\cell \in \partition \setminus \partition_b} \eucl{\centerof\cell}{b} \hat{\sigma}(\centerof\cell,b)\right] \\
 &\leq \varepsilon \sum_{b \in B} \left[\int_{\mathcal{X}} \eucl{a}{b} \hat{\tau}(a,b) \, da 
 + (1+\varepsilon)^2 \int_{\mathcal{Z}_b} \eucl{a}{b} \hat{\tau}(a,b) \, da\right] \\
 &\leq \frac{9}{4} \varepsilon \sum_{b \in B} \left[\int_{\mathcal{X}} \eucl{a}{b} \hat{\tau}(a,b) \, da 
 + \int_{\mathcal{Z}_b} \eucl{a}{b} \hat{\tau}(a,b) \, da\right],
 \end{align*}
 where the second and fourth lines follow from the first condition of Lemma \ref{lem:partition-properties}, the third line follows from Lemma \ref{lemma:approx-largedists} and the fact that $\sigma$ is a $(1+\varepsilon)$-approximate transport plan, and the last line uses $\varepsilon \in (0, \frac{1}{2})$. We conclude that $\plancost(\tau) \leq (1+\frac{9}{4}\varepsilon) \plancost(\hat{\tau})$.
\end{proof}

\section{Missing Proofs of Section \ref{sec:discrete-OT}}

\discretegraph*

\begin{proof}
The vertex set of our graph consists of the center points of all non-empty cells of the quad-tree as well as the point sets $A\cup B$. At each level $i$ of the tree, the total number of non-empty cells of level $i$ is no more than $n$. Since our spanner contains the center point of each non-empty cell at all levels, where $h=O(\log \log n)$, the total number of vertices is $O(n\log \log n)$.

Next, we bound the number of edges of our graph. For any cell $\cell$, we add two sets of edges corresponding to two $(1+\varepsilon)$-spanners $\mathcal{S}_\cell$ and $\mathcal{S}'_\cell$. Each one of these spanners has $O(n_\cell \varepsilon^{-d})$ edges and bounded degree of $O(n_\cell \varepsilon^{-d} \log n) \leq O(n \varepsilon^{-d} \log n)$ for cell $\cell$. In each level of the graph, there are at most $n$ points distributed among cells where each point appears at most once. Therefore, $|E|$ is bounded by a sum over all levels $\ell$ of the graph:
\[O(\sum_{\cell} n_\cell \varepsilon^{-d}) \leq O(\sum_{\ell} n \varepsilon^{-d}) = O(n \varepsilon^{-d} h).\]

The cost of any edge in the spanner $\spanner$ is the Euclidean distance of the two endpoints of the edge. Therefore, from the triangle inequality, any path from $a$ to $b$ has a cost of at least the Euclidean distance of $a$ and $b$; i.e, $\lengthof{\pathof{a,b}}\ge \eucl{a}{b}$.

Suppose $(a,b)$ have least common ancestor $\cell$. Let $\subcell_a, \subcell_b$ be the subcells in $\mathcal{S}'_\cell$ which contain $a$ and $b$, respectively. Since $\mathcal{S}'_\cell$ is a $(1+\varepsilon)$-spanner, the length of the shortest path from $\subcell_a$ to $\subcell_b$ is a $(1+\varepsilon)$-approximation of their Euclidean distance.

Define $\pathof{a}$ and $\pathof{b}$ to be the shortest paths from $a$ to $\subcell_a$ and $b$ to $\subcell_b$, respectively, only taking greedy edges. Then, one path from $a$ to $b$ in the graph is the following path:
\[ P = \pathof{a} \circ \pathofcell{a,b}{\cell} \circ  \pathof{b}.  \]
For any cell $\cell$, define $\subcelldiam{\cell} = \sqrt{d}\varepsilon h^{-1} \ell_{\cell}$ to be the diameter of the subcells of $\cell$. Define $\subcelldiam{a}$ and $\subcelldiam{b}$ to be the diameter of the subcells $\subcell_a$ and $\subcell_b$. Recall that $\cell$ is of level $i>0$.

For any $\cell' \in \children\cell$, we note that the shortest path from $\cell'$ to $\cell$ is bounded above in length by $(1+\varepsilon) \eucl{c_{\cell'}}{c_\cell}$. Using this, we can bound the length of $\pathof{a}$ and $\pathof{b}$ by the greedy paths going directly up the tree:

\[ \lengthof{\pathof{a}} + \lengthof{\pathof{b}} \le \frac{1}{2}(\subcelldiam{a} + \subcelldiam{b}).  \]
Next, we bound the expected value of $\subcelldiam{a}$ and $\subcelldiam{b}$. For any level $j$ of the tree, the probability that the least common ancestor of $(a,b)$ is of level $j$ is 
\[\prob{\level{a,b}=j}\le \frac{\sqrt{d}\eucl{a}{b}}{\ell_{j+1}}.\]
As a result,
\begin{equation*}
    \expect{\subcelldiam{a}} \le \sum_{j=1}^{h} \prob{\level{a,b}=j}.\subcelldiam{j+1} \le \sum_{j=1}^{h}\frac{\sqrt{d}\eucl{a}{b}}{\ell_{j+1}}.\frac{\varepsilon \ell_{j+1}}{2\sqrt{d}h} =\frac{\varepsilon}{2}\eucl{a}{b} .
\end{equation*}
An analogous claim can be made for $\subcelldiam{b}$. Finally, as discussed before, the cost of the shortest path between $c_{\subcell_a}, c_{\subcell_b}$ is bounded above by $ (1+\varepsilon)\eucl{\centerof{\subcell_a}}{\centerof{\subcell_b}}$. Using triangle inequality,
\[ \eucl{\centerof{\subcell_a}}{\centerof{\subcell_b}}\le \eucl{a}{b} + \frac{1}{2}(\subcelldiam{a} + \subcelldiam{b}) \]
Combining all these bounds,
\begin{align*}
    \expect{\lengthof{\pathof{a,b}}}&\le \expect{\lengthof{P}} =  \expect{\lengthof{\pathof{a}}+\lengthof{\pathofcell{ab}{\cell}}+\lengthof{\pathof{b}}}\\ &\le \expect{(1+\varepsilon)\eucl{\centerof{\subcell_a}}{\centerof{\subcell_b}} + \frac{1}{2}(\subcelldiam{a} + \subcelldiam{b})}\\ &\le (1+\varepsilon)\eucl{a}{b} + \frac{1}{2}((1+\varepsilon)+1)\expect{\subcelldiam{a} + \subcelldiam{b}}\\ &\le \left((1+\varepsilon) + \frac{(2+\varepsilon)\varepsilon}{2}\right)\eucl{a}{b}\le (1+\frac{5}{2}\varepsilon)\eucl{a}{b}.
\end{align*}
where the last inequality assumes $\varepsilon \leq 1$. If not, then $\varepsilon$ can be substituted for $1$ without loss of generality. To obtain $(1+\varepsilon)$-approximation instead, one can rescale $\varepsilon$ by $\frac{2}{5}$.

\end{proof}

\improveddualdistortion*

\begin{proof}
For any edge $(u,v)\in E$, consider the following cases.

\begin{enumerate}
    \item \textit{Greedy edges:} If $(u,v)$ is an greedy edge, by the definition, there exists a cell $\cell$ such that $u,v \in \localinstance_\cell$. Let $(f_\cell,y_\cell)$ denote the flow and the set of dual weights computed on the local instance $\localinstance_\cell$. From the properties of exact primal-dual minimum cost flow, $|y_\cell(u) - y_\cell(v)| \le \eucl{u}{v}$. Therefore, by the dual assignment of our algorithm,
    \[ |y(u) - y(v)| = |(y_\cell(u)-y_\cell(\centerof{\cell}) + y(\centerof\cell))-(y_\cell(v)-y_\cell(\centerof{\cell}) + y(\centerof\cell))| \le \eucl{u}{v}. \]

    \item \textit{Shortcut edges:} If $(u,v)$ is a shortcut edge, then there exists a cell $\cell$ of level $i$ and children $\cell_1, \cell_2\in\children\cell$ such that $u$ (resp, $v$) is the center point of a subcell $\subcell_1\in \subcells{\cell_1}$ (resp. $\subcell_2\in \subcells{\cell_2}$); i.e, $u=\centerof{\subcell_1}$ (resp. $v=\centerof{\subcell_2}$). Observe that $\centerof{\subcell_1} \in \localinstance_{\cell_1}$ and $\centerof{\subcell_2} \in \localinstance_{\cell_2}$. Recall that $\spannercell{\cell_1}$ (resp. $\spannercell{\cell_2}$) denotes the $(1+\varepsilon)$-spanner constructed on the local instance $\localinstance_{\cell_1}$ (resp. $\localinstance_{\cell_2}$). Let $P = \langle \centerof{\subcell_1}=p_1, \dots, p_{k_1}=\centerof{\cell_1} \rangle$ be the path in $\spannercell{\cell_1}$ from $\centerof{\subcell_1}$ to $\centerof{\cell_1}$. Similarly, let $Q = \langle \centerof{\subcell_2}=q_1, \dots, q_{k_2}=\centerof{\cell_2} \rangle$ be the path in $\spannercell{\cell_2}$ connecting $\centerof{\subcell_1}$ to $\centerof{\cell_1}$. Finally, note that  $\centerof{\cell_1}, \centerof{\cell_2} \in \localinstance_\cell$ and let $R=\langle \centerof{\cell_1} = r_1, \dots, r_{k_3} = \centerof{\cell_2} \rangle$ be the path connecting the two center points $\centerof{\cell_1}$ and $\centerof{\cell_2}$ in $\spannercell{\cell}$. All the edges in the paths $P, Q,$ and $R$ are greedy edges. 
    By the triangle inequality,
    \begin{align*}
        |y(\centerof{\subcell_1}) - y(\centerof{\subcell_2})| &\leq |y(\centerof{\subcell_1}) - y(\centerof{\cell_1})| + |y(\centerof{\cell_1}) - y(\centerof{\cell_2})| + |y(\centerof{\cell_2}) - y(\centerof{\subcell_2})| \\
        &\leq \left( \sum_{j=1}^{k_1-1} |y(p_j) - y(p_{j+1})|\right) + \left( \sum_{j=1}^{k_3-1} |y(r_i) - y(r_{i+1})| \right) \\
        &\;\;\;\;\;+ \left( \sum_{j=1}^{k_2-1} |y(q_{j+1}) - y(q_j)| \right) \\
        &\leq (1+\varepsilon) \eucl{\centerof{\subcell_1}}{\centerof{\cell_1}} + (1+\varepsilon) \eucl{\centerof{\cell_1}}{\centerof{\cell_2}} + (1+\varepsilon) \eucl{\centerof{\cell_2}}{\centerof{\subcell_2}}.
    \end{align*}
    
    Since $\subcell_1$ and $\subcell_2$ are subcells of children $\cell_1$ and $\cell_2$ of $\cell$, their side-lengths are both $\frac{\eps\sidelength{\cell_1}}{2dh}$. Thus, the Euclidean distance of their centers is $\eucl{\centerof{\subcell_1}}{\centerof{\subcell_2}}\ge\frac{\eps\sidelength{\cell_1}}{2dh}$. Furthermore, $\eucl{\centerof{\subcell_1}}{\centerof{\cell_1}} \leq \sqrt{d} \sidelength{\cell_1}$ and $\eucl{\centerof{\cell_2}}{\centerof{\subcell_2}} \leq \sqrt{d} \sidelength{\cell_2}$. Combining these inequalities gives
    \[\eucl{\centerof{\subcell_1}}{\centerof{\cell_1}} \leq \sqrt{d} \sidelength{\cell_1} \leq \frac{2d^{3/2} h}{\varepsilon} \eucl{\centerof{\subcell_1}}{\centerof{\subcell_2}},\]
    and the analogous for $\eucl{\centerof{\cell_2}}{\centerof{\subcell_2}}$. By triangle inequality, we can extend this to conclude $\eucl{\centerof{\cell_1}}{\centerof{\cell_2}} \leq O(\frac{d^{3/2}h}{\varepsilon})\eucl{\centerof{\subcell_1}}{\centerof{\subcell_2}}$. Therefore, 
    \begin{align*}
        |y(\centerof{\subcell_1}) - y(\centerof{\subcell_2})| &\leq (1+\varepsilon) \eucl{\centerof{\subcell_1}}{\centerof{\cell_1}} + (1+\varepsilon) \eucl{\centerof{\cell_1}}{\centerof{\cell_2}} + (1+\varepsilon) \eucl{\centerof{\cell_2}}{\centerof{\subcell_2}} \\
        &\leq O\left(\frac{d^{3/2}h}{\eps}\right)\eucl{\centerof{\subcell_1}}{\centerof{\subcell_2}}.
    \end{align*}
\end{enumerate}
\end{proof}

\strongduality*
\begin{proof}
    By construction, for any shortcut edge $(u,v)\in E$, $\sigma(u,v) > 0$. For any greedy edge $(u,v)\in E$, there exists a unique cell $\cell$ such that $u,v\in \localinstance_\cell$ and the spanner $\spannercell{\cell}$ contains the edge $(u,v)$. By the dual assignment, if $\sigma_\cell(u,v)>0$, then
    \[y(u) - y(v) = y_\cell(u) - y_\cell(v) = \eucl{u}{v}.\]
    Therefore, for any edge $(u,v)$ carrying a positive flow in $\sigma$, $y_\cell(u) - y_\cell(v) = \eucl{u}{v}$. As a result,
    \begin{align*}
    \sum_{(u,v) \in E} \sigma(u,v) \eucl{u}{v} &= \sum_{(u,v) \in E} \sigma(u,v) (y(u) - y(v)) \\
    &= \sum_{w \in V} \left(\sum_{z : (w,z) \in E} \sigma(w,z)\right) y(w) \\
    &= \sum_{w \in V} y(w) \cdot \eta(w).
\end{align*}
\end{proof}

\section{The Multiplicative Weight Update Framework}
\label{sec:boosting}

At a very high level, the multiplicative weights method uses an approximate oracle to estimate the best flow and iteratively updates the flow using the oracle as a rough guide. In our setting, we construct an undirected graph $\spanner = (V,E)$ with $A \cup B \subseteq V$ that is a (randomized) spanner of $A \cup B$ under Euclidean distance and we compute an $\varepsilon$-approximate MCF on $\spanner$. The approximate oracle is a greedy algorithm $\greedy$ which routes flow along tree edges. This greedy tree flow leads to high costs for some pairs which have positive flow, and the multiplicative weights method gradually reroutes the flow along shorter paths between these two points in the graph.

We use complementary slackness to guide which edges are valuable. Using the LP duality, the MCF problem can be formulated as computing dual weights $y$ maximizing $\sum_{v \in V} \eta(u) y(u)$ subject to $y(u) - y(v) \leq \eucl{u}{v}$ for all $(u,v) \in E$. Equivalently, the collection of constraints can be expressed as $\max_{(u,v) \in E} \frac{y(u) - y(v)}{\eucl{u}{v}} \leq 1$. We refer to the expression $\frac{y(u) - y(v)}{\eucl{u}{v}}$ as the \emph{slack} of an edge $(u,v)$. By complementary slackness, if $(\sigma,y)$ is an optimal primal-dual pair, then $\sigma(u,v)$ is positive when the slack of $(u,v)$ is 1. In view of this observation, if the slack of a directed edge $e$ is large, the MWU method increases the flow along $e$.

We transform the undirected graph $\spanner$ into a directed graph $G=(V,E)$ that takes the vertices of $\spanner$ and adds both directed edges for every undirected edge of $\spanner$. Additionally set $\eta(u)$ to be the original demand of $u \in A \cup B$ and $0$ for $u \in V \setminus (A \cup B)$. The cost of an edge $e=(u,v)$ in $G$, denoted by $\edgecost{e}$, is $\eucl{u}{v}$. Given $G$, a demand function $\eta \colon V \to \real$, and a parameter $\varepsilon > 0$, the MWU algorithm computes a flow function $\sigma \colon E \to \real_{\geq 0}$ that satisfies the demand and $\plancost(\sigma) \leq (1+\varepsilon) \opt$, where $\opt$ is the min-cost flow for $(G, \eta)$. The algorithm assumes the existence of a greedy algorithm $\greedy(G, \eta)$ that computes a primal-dual pair $(\sigma, y)$ on $G$, where $\sigma \colon E \rightarrow \real$ is a flow function that routes the demand $\eta$, i.e. $\sum_{v : (u,v) \in E} \sigma(u,v) - \sigma(v,u) = \eta(u)$ for all $u \in V$, and $y \colon V \to \real$ is a dual weight function that satisfies the following two conditions:

\smallskip
\begin{description}
    \item[(C1)] $|y(u) - y(v)| \leq \rho \eucl{u}{v}$ \;\; $\forall (u,v)\in E$,
    \smallskip
    \item[(C2)] $\sum_{(u,v) \in E} \sigma(u,v) \eucl{u}{v} \leq \sum_{u \in V} y(u) \eta(u)$,
\end{description}

\smallskip
where $\rho > 0$ is a parameter. (C1) guarantees $\rho$-approximate feasibility of the computed dual weights. (C2) is a strong-duality condition used to prevent $\greedy$ from returning trivial dual weights, as well as upper bound the cost of the flow $\sigma$.

We now describe the algorithm in more detail. Using one of the known algorithms \cite{charikar2002similarity, indyk2003fast}, we first compute an estimate of the OT cost within a $d \log n$ factor in $O(n \log n)$ time, i.e. it returns a value $\tilde{g}$ such that $\opt \leq \tilde{g} \leq (d\log n) \cdot \opt$. We refer to this algorithm as {\scshape LogApprox}. Using this estimate, we perform an exponential search in the range $\left[\frac{\tilde{g}}{d\log n}, \tilde{g}\right]$ with increments of factor $(1+\varepsilon)$. For any guess value $g$, the MWU algorithm either returns a flow $\sigma \colon E \to \real$ with $\plancost(\sigma) \leq (1+\varepsilon) g$ or returns dual weights as a certificate that $g < \opt$. We now describe the MWU algorithm for a fixed value of $g$.

Set $T = 4\rho^2 \varepsilon^{-2} \log |E|$. The algorithm runs in at most $T$ iterations, where in each iteration, it maintains a pre-flow vector $\sigma^t$ such that $\plancost(\sigma^t) \leq g$. The flow $\sigma^t$ need not route all demand successfully. Initially, set $\sigma^0(e) = \frac{g}{\edgecost{e} \cdot |E|}$ so that $\plancost(\sigma^0) \leq g$. For each iteration $t$, define the \emph{residual demand} of iteration $t$, denoted by $\eta_{\text{res}}^t$, as 
\[\eta_{\text{res}}^t(u) = \eta(u) - \sum_{v : (u,v) \in E} (\sigma^{t-1}(u,v) - \sigma^{t-1}(v,u)).\]
Let $(\sigma_{\text{res}}^t, y^t)$ be the primal-dual flow computed by the $\greedy$ for the residual demands $\eta_{\text{res}}^t$. Recall that $(\sigma_{\text{res}}^t, y^t)$ satisfies (C1) and (C2).
If $\langle \eta_{\text{res}}^t, y^t \rangle \leq \varepsilon g$, then (C2) implies that $\plancost(\sigma_{\text{res}}^t) \leq \varepsilon g$. Since $\sigma_{\text{res}}^t$ routes the residual demands, the flow function $\sigma^t = \sigma^{t-1} + \sigma_{\text{res}}^t$ routes the original demand $\eta$ with a cost $\plancost(\sigma^t) \leq (1+\varepsilon) g$. In this case, the algorithm returns $\sigma^t$ as the desired flow and terminates.

Otherwise, $\langle \eta_{\text{res}}^t, y^t \rangle > \varepsilon g$ and we update the flow along each edge $e=(u,v)$ of $G$ based on the slack $\slack^t(u,v) = \frac{y^t(u) - y^t(v)}{\eucl{u}{v}}$ of $e$ with respect to dual weights $y^t$:
\[\sigma^t(u,v) \leftarrow \exp\left(\frac{\varepsilon}{2\rho^2} \slack^t(u,v)\right) \cdot \sigma^{t-1}(u,v).\]

We emphasize that flow along an edge is increasing if the slack is large. Then, one needs to rescale $\sigma^t$ so that its cost is bounded above by $g$. If the algorithm does not terminate within $T$ rounds, we conclude that the value of $g$ is smaller than the MCF cost. We increase $g$ by a factor of $(1+\varepsilon)$ and repeat the MWU algorithm.

\begin{algorithm}
    \caption{Minimum-Cost-Flow$(G, \eta, \varepsilon)$:}
    \label{alg:boost}
    \begin{algorithmic}
        \STATE $\tilde{g} \leftarrow \textsc{LogApprox}(G, \eta)$
        \STATE $g \leftarrow \frac{\tilde{g}}{d\log n}, T = 4\rho^2 \varepsilon^{-2} \log |E|$
        \REPEAT
        \STATE $\sigma^0(u,v) = \sigma^0(v,u) = \frac{g}{\eucl{u}{v} \cdot |E|} \,\, \forall (u,v) \in E$
        \FOR{$t = 1, \dots, T$}
        \STATE $\eta_{\text{res}}^t(u) \leftarrow \eta(u) - \left(\sum_{v : (u,v) \in E} \big[\sigma^{t-1}(u,v) - \sigma^{t-1}(v,u)\big]\right) \,\, \forall u \in V$
        \STATE $(\sigma_{\text{res}}^t, y^t) \leftarrow \greedy(G, \eta_{\text{res}}^t)$
        \IF{$\langle \eta_{\text{res}}^t, y^t \rangle \leq \varepsilon \cdot g$}
        \RETURN $\sigma^{t-1} + \sigma_{\text{res}}^t$
        \ENDIF
        \STATE $\slack^t(u,v) \leftarrow \frac{y^t(u) - y^t(v)}{\eucl{u}{v}} \quad \forall(u,v) \in E$
        \STATE $\sigma^t(u,v) \leftarrow \exp(\beta \slack^t(u,v)) \cdot \sigma^{t-1}(u,v)$
        \STATE Rescale $\sigma^t$ so $\plancost(\sigma^t) \leq g$
        \ENDFOR
        \STATE $g \leftarrow (1+\varepsilon) g$
        \UNTIL{flow is returned}
    \end{algorithmic}
\end{algorithm}

The following Lemma regarding Algorithm \ref{alg:boost} is proven in \cite{zuzic2021simple}, which we follow closely in this work for its use of primal-dual oracles.

\begin{lemma}\label{lemma:boosting}
    Given an $O(d\log n)$ approximate guess $g$ of the minimum cost flow value and an algorithm $\greedy$ which computes a primal-dual pair satisfying conditions (C1) and (C2) in $T_\rho(n)$ time, a $(1+\varepsilon)$-approximate minimum cost flow on $G$ can be computed in $O((T_\rho(n) + |E|) \frac{\rho^2 \log n}{\varepsilon^2} \log (d\log n))$ time.
\end{lemma}

\subsection{Recovering a Transport Map}

To be precise, the multiplicative weights algorithm we have described so far produces a min-cost flow on the $(1+\varepsilon)$-spanner in expectation. A true transportation map is over $A \times B$. We briefly describe the procedure of \cite{khesin2019preconditioning} for completeness, which takes a flow on some approximate spanner with bounded degree and produces a transportation map. The basic idea is to iteratively skip over any vertex which has flow passing through. Algorithm \ref{alg:recover} shows how to shortcut vertices.

\begin{algorithm}
    \caption{Recover-Transport-Map$(G, A, B, \tau)$:}
    \label{alg:recover}
    \begin{algorithmic}
        \STATE Let $f(u,v) = \tau(u,v) - \tau(v,u)$
        \WHILE{there exists $v \in V \setminus (A \cup B)$ where $f(u, v), f(v,w) > 0$ for some $u,w$}
        \STATE Add $(u,w)$ to $G$ if $(u,w) \not \in E$
        \STATE $f(u,w) \leftarrow \min\{f(u,v), f(v,w)\}$
        \STATE Subtract $\min\{f(u,v), f(v,w)\}$ from both $f(u,v)$ and $f(v,w)$
        \STATE Remove $(u,v)$ from $G$ if $f(u,v) = 0$ (analogous for $(v,w)$)
        \IF{deg$(u) > $ deg$(v)$ and $u \not \in A \cup B$ (analogous for $w$)}
        \STATE Shortcut $u$ (or $w$) in next iteration of while loop
        \ENDIF
        \ENDWHILE
    \end{algorithmic}
\end{algorithm}

\begin{lemma} \label{lemma:recover}
    Given a graph $G = (V,E)$ with $A \cup B \subseteq V$ and maximum degree of $\text{deg}_{\text{max}}$, as well as a flow $\tau$ over $G$ which routes the demand of $A \cup B$, Algorithm \ref{alg:recover} returns a transportation plan over $A \times B$ in $O(|E| \cdot \text{deg}_{\text{max}})$ time.
\end{lemma}

\end{document}